\documentclass[10pt]{book}
\usepackage{array}
\usepackage{float}
\usepackage{mathrsfs}
\usepackage{amsmath}
\usepackage{amssymb}
\usepackage{graphicx}
\usepackage{amsthm}
\usepackage{amsfonts}
\usepackage{epsfig}
\usepackage{bm}
\usepackage{enumerate}
\usepackage{url}

\floatstyle{ruled}
\newfloat{algorithm}{tbp}{loa}
\providecommand{\algorithmname}{Algorithm}
\floatname{algorithm}{\protect\algorithmname}

\newtheorem{definition}{Definition}[chapter]
\newtheorem{theorem}[definition]{Theorem}
\newtheorem{prop}{Proposition}[chapter]

\newtheorem{lemma}[definition]{Lemma}

\newcounter{hypB}
\newenvironment{hypB}{\refstepcounter{hypB}\begin{itemize}
  \item[({\bf B\arabic{hypB}})]}{\end{itemize}}

\begin{document}


\renewcommand{\baselinestretch}{1.2}

\markright{\hbox{\footnotesize\rm  }\hfill}

\markboth{\hfill{\footnotesize\rm ADAM PERSING AND AJAY JASRA} \hfill}
{\hfill {\footnotesize\rm TWISTING THE ALIVE PARTICLE FILTER} \hfill}

\renewcommand{\thefootnote}{}
$\ $\par


\fontsize{10.95}{14pt plus.8pt minus .6pt}\selectfont
\vspace{0.8pc}
\centerline{\large\bf TWISTING THE ALIVE PARTICLE FILTER}
\vspace{.4cm}
\centerline{Adam Persing$^{1}$ and Ajay Jasra$^{2}$}
\vspace{.4cm}
\centerline{\it $^{1}$Department of Mathematics, Imperial College London}
\centerline{\it $^{2}$Department of Statistics \& Applied Probability, National University of Singapore}
\vspace{.55cm}
\fontsize{9}{11.5pt plus.8pt minus .6pt}\selectfont


\begin{quotation}
\noindent {\it Abstract:}
This work focuses on sampling from hidden Markov models \cite{Cappe_2005} whose observations have intractable density functions.  We develop a new sequential Monte Carlo (\cite{Doucet_2008}, \cite{Doucet_2001}, \cite{Gordon_1993}) algorithm and a new particle marginal Metropolis-Hastings \cite{Andrieu_2010} algorithm for these purposes.  We build from \cite{Jasra_2013} and \cite{Whiteley_2013} to construct the sequential Monte Carlo (SMC) algorithm (which we call the alive twisted particle filter).  Like the alive particle filter of \cite{Jasra_2013}, our new SMC algorithm adopts an approximate Bayesian computation \cite{Tavare_1997} estimate of the HMM.  Our alive twisted particle filter also uses a twisted proposal as in \cite{Whiteley_2013} to obtain a low-variance estimate of the HMM normalising constant.  We demonstrate via numerical examples that, in some scenarios, this estimate has a much lower variance than that of the estimate obtained via the alive particle filter.  The low variance of this normalising constant estimate encourages the implementation of our SMC algorithm within a particle marginal Metropolis-Hastings (PMMH) scheme, and we call the resulting methodology ``alive twisted PMMH''.  We numerically demonstrate on a stochastic volatility model how our alive twisted PMMH can converge faster than the standard alive PMMH of \cite{Jasra_2013}.
\par

\vspace{9pt}
\par
\end{quotation}\par


\fontsize{10.95}{14pt plus.8pt minus .6pt}\selectfont

\setcounter{chapter}{1}
\setcounter{section}{1}
\setcounter{subsection}{0}
\setcounter{figure}{0}
\setcounter{equation}{0} 
\noindent {\bf 1. Introduction}

Consider a Markov process that evolves in discrete time, and assume that we cannot directly observe the states of the process.  At each time point, we can only indirectly observe the latent state through some other random variable.  Assume also that the observation random variables are statistically independent of one another, conditioned upon the latent Markov process at the current state.  This model is called a hidden Markov model \cite{Cappe_2005}.  Hidden Markov models (HMMs) are very flexible, and so they are used in a wide array of real world applications.   Some examples include stochastic volatility models \cite{Langrock_2012}, real-time hand writing recognition \cite{Hu_1996}, and DNA segmentation \cite{Peshkin_1999}.

As increasing amounts of data have become available to practitioners, real world HMMs have become more and more complex; see \cite{Colella_2007}, \cite{Rosales_2004} and \cite{Yau_2011} for examples.  In instances where analytical Bayesian inference for an HMM is not feasible, one may resort to numerical methods, such as the extended Kalman filter \cite{Anderson_1979}, the unscented Kalman filter \cite{Julier_2004}, block updating Markov chain Monte Carlo \cite{Shephard_1997bb}, sequential Monte Carlo (\cite{Doucet_2008},\cite{Doucet_2001},\cite{Gordon_1993}), and particle Markov chain Monte Carlo \cite{Andrieu_2010}.  The latter two methods are widely regarded as the state-of-the-art, and we further discuss the details of those techniques in Sections 4 and 8 below.  We briefly say here that sequential Monte Carlo (SMC) is a popular class of algorithms that simulate a collection of $N$ weighted samples (or ``particles'') of the HMM's hidden state sequentially in time by combining importance sampling and resampling techniques.  SMC algorithms can be implemented online, and they can be used to obtain unbiased estimates of the probability density of the HMM's observations given the model parameters (i.e., the marginal likelihood).  Particle marginal Metropolis-Hastings (PMMH) algorithms (a type of particle Markov chain Monte Carlo) employ SMC and an unbiased estimate of the marginal likelihood within a Metropolis-Hastings scheme to sample the latent states of the HMM and the model parameters (should they also be unknown).  In \cite{Andrieu_2010}, the authors explain that the variance of the unbiased estimate of the marginal likelihood is critical in the performance of PMMH.  Also note that PMMH is not an online methodology.  Both SMC and PMMH have applications outside of HMMs, but we do not explore those applications here.


Most SMC and PMMH techniques require calculation of the likelihood density of the observations given the latent state (see Sections 4 and 8 below).  The likelihood density can be quite difficult (or even impossible) to evaluate for some complex models, and the focus of this paper is on the subset of HMMs whose observations have unknown or intractable likelihood densities.  Such models can arise, for example, when modeling stochastic volatility \cite{Jasra_2013} and optimising portfolio asset allocation \cite{Jasra_2012b}.  It is typical to adopt approximate Bayesian computation (ABC) \cite{Tavare_1997} estimates of these models (thereby replacing the likelihood density with an approximation) and then use SMC techniques such as in \cite{Jasra_2012b} and \cite{Martin_2012} to perform inference.  The SMC methods of \cite{Jasra_2012b} and \cite{Martin_2012} can, in practice, yield samples that all have a weight of zero (i.e., they can die out).  The competing methods of \cite{Cerou_2012} and \cite{Del_Moral_2012} (which can be used outside of the context of HMMs) reduce the possibility of dying out, but even those methods are not guaranteed to work in practice.  The alive particle filter and the alive PMMH algorithm of \cite{Jasra_2013}, on the other hand, are better methods both in practice and in theory because they cannot die out.

The goal of this paper is to develop new SMC and PMMH algorithms, for ABC approximations of HMMs, that cannot die out and further improve over existing methodologies.  To that end, this paper improves upon the work of \cite{Jasra_2013} by twisting the proposals of the alive algorithms; in a twisted SMC algorithm \cite{Whiteley_2013}, a change of measure is applied to a standard SMC algorithm to reduce its variance (see Section 4 below for a review).  We introduce here a change of measure to the alive algorithms of \cite{Jasra_2013} to yield new ``alive twisted'' sampling techniques that cannot die out and (under certain scenarios) have a superior performance compared to the SMC and PMMH methods of \cite{Jasra_2013}.

The paper commences by introducing our notation in Section 2 (the notation is similar to the Feynman-Kac notation of \cite{Whiteley_2013}).  Section 3 provides a review of the HMMs of interest. Section 4 provides a brief review of SMC and the particular algorithms from which our work builds (i.e., the alive particle filter and the twisted particle filter).  Section 5 states our new alive twisted SMC algorithm.  In Section 6, we state the optimal change in measure of our new algorithm; the assumptions and the main theorem that justify this change of measure are stated as well.  A proof of this result is given in Section A of the appendix, and it follows the framework developed in \cite{Whiteley_2013}.  We numerically compare the alive twisted SMC algorithm to the alive particle filter in Section 7 by implementing both on an ABC approximation of a linear Gaussian HMM.  The numerical example shows that the new algorithm has a significantly lower variance under certain scenarios.  The encouraging performance of the alive twisted particle filter prompts us to embed it within a PMMH algorithm in Section 8, and we find empirically that the alive twisted PMMH is able to converge faster than a non-twisted alive PMMH in Section 9 (when both algorithms are implemented on a stochastic volatility model).  Section 10 concludes the paper with a summary and discussion.
\par

\setcounter{chapter}{2}
\setcounter{section}{2}
\setcounter{subsection}{0}
\setcounter{figure}{0}
\setcounter{equation}{0} 
\noindent {\bf 2. Notation and definitions}

Consider a random variable $X_{k}\in\mathbb{R}^{d_{x}}$ (with index $k$) which may take a value $x_{k}$.  The vector of all $X_{k}$'s corresponding to all $k\in\{1,\dots,n \}$ for $n \geq 1$ will be designated $X_{1:n}$, and the joint density of $X_{1:n}$ will be written $\pi_{\theta}\left( x_{1:n} \right)$; conditional densities of the form $\pi_{\theta} \left( x_{n} \mid x_{n-1} \right)$ may sometimes be written as $\pi_{\theta} \left( x_{n-1}, x_{n} \right)$.  When $X_k$ is to be drawn from the distribution corresponding to the density $\pi_{\theta}\left( x_{k} \right)$, we will slightly abuse the notation and write $X_k \sim \pi_{\theta}\left( \cdot \right)$.  The parameter $\theta\in\Theta\subseteq\mathbb{R}^{d_{\theta}}$ is static across all values of $n$; $\mathscr{B}(\Theta)$ denotes the Borel sets on $\Theta$ and $\mathscr{P}(\Theta)$ denotes the class of probability measures on $\Theta$.  Standard notation $\mathbb{E}[ X_k ]$ is adopted for the expected value of a random variable $X_k$, and similarly, the variance of the same random variable is denoted $\mathbb{V}[ X_k ]$.  All distributions used in this work can be found in Table \ref{table:dist} in Section B of the appendix.

It is assumed that each joint density $\pi_{\theta}\left( x_{1:n} \right)$ may be decomposed as $\pi_{\theta}\left( x_{1:n} \right) = {\gamma_{\theta}\left( x_{1:n} \right)}/{Z_{\theta,1:n}}$, where $Z_{\theta,1:n}$ is a normalising constant.  When a density conditions on a sequence of random variables (and that sequence of random variables is obvious), we may interchangeably use $\pi_{\theta}\left(x_1\mid x_{2:n}\right)=\pi_{\theta}\left( x_1 \mid \cdots \right)$.  Any approximations to any density $\pi_{\theta}\left( x_{1:n} \right)$ will be denoted $\widehat{\pi}_{\theta}\left( x_{1:n} \right)$, and the approximation of the normalising constant will be similarly written as $\widehat{Z}_{\theta,1:n}$.  A collection of $N$ samples with a common density $\pi_{\theta}\left( x_{1:n} \right)$ will be written as $x_{1:n}^{1:N}$ or $\underline{x}_{1:n}$.  We will sometimes assign the value $a_k^i \in \{1,\dots,N\}$ to be the index of a sampled value for $X_k$, and in that instance, we allow $x_k^{a_k^i}=x_k^{a(i)}$.  A sequence of index assignments $a_1^i$, $a_2^i$, $\dots$, $a_n^i$ will be written as $a_{1:n}^i$.

A probability space $(\Omega,\mathscr{F},\mathbb{P}_{\theta})$ consists of a sample space $\Omega$ and a set of events $\mathscr{F}$.  $\mathbb{P}_{\theta}$ is a probability measure defined for every $\theta\in\Theta$ such that for every $A\in\mathscr{F}$, $\mathbb{P}_{\theta}(A)$
is $\mathscr{B}(\Theta)$-measurable.  Furthermore, $\mathscr{M}(\Omega)$ denotes the collection of measures on $\Omega$ and $\mathscr{P}(\Omega)$ can also denote the collection of probability measures on $\Omega$.  When clearly stated, $\mathscr{F}$ may alternatively denote a filtration.

The conventions $\sum_{i=n}^{n-1} = 0$ and $\prod_{\emptyset} = 1$ are used throughout, and $a \wedge b$ denotes the minimum between the two real numbers $a$ and $b$.  For the real-valued numbers $(a,b)$, $\mathbb{I}_{a}\left( b \right)$ will denote the indicator function that equals one when $a=b$ and zero when $a \neq b$.  For some measurable space $(E,\mathcal{E})$, let $B_\epsilon \left(a\right)\subseteq \mathcal{E}$ be a ball of radius $\epsilon$ centred on $a\in E$.  Thus, for $b\in E$, $\mathbb{I}_{B_{\epsilon} \left(a\right)} \left( b \right)$ will be an indicator function whose value is one when $b\in B_{\epsilon} \left(a\right)$ and zero otherwise.

For a measurable function $\varphi:\mathbb{R}^d\rightarrow\mathbb{R}$ such that $\sup_{x\in\mathbb{R}^d}|\varphi(x)|<+\infty$, we write $\varphi\in\mathcal{B}_b(\mathbb{R}^d)$.  $\mathcal{B}_b(\mathbb{R}^d)$ is the Banach space that is complete with respect to the norm $\sup_{x\in\mathbb{R}^d}|\varphi(x)|$.

As our work follows that of \cite{Whiteley_2013}, we use as similar a notation as possible to that original article.  Consider a sequence of $(H,\mathcal{H})$-valued random variables, denoted by $Y_{n}\in\mathbb{R}^{d_{y}}$ at each (time) point $n$.  For the probability space $(\Omega,\mathscr{F},\mathbb{P})$, let $\Omega=H^{\mathbb{Z}}$ be the set of doubly infinite sequences valued in $H$ and let $\mathscr{F}=\mathcal{H}^{\otimes \mathbb{Z}}$.  Then for $\omega=\{\omega\left(n\right)\}_{n\in\mathbb{Z}}\in\Omega$, we can write each random variable as $Y_n = Y_n\left(\omega\right)=\omega\left(n\right)$.  We use $z$ to define a shift operator $z:\Omega \rightarrow \Omega$ as $\left(z\omega\right)\left(n\right)=\omega\left(n+1\right)$, where applying $z$ $m$-times is written $z^m$.  Thus, for example, $Y_n\left(z\omega\right)=Y_{n+1}\left(\omega\right)$ and $Y_n\left(z^m\omega\right)=Y_{n+m}\left(\omega\right)$.

Now consider two more evolving discrete time processes: the $(R,\mathcal{R})$-valued process with states denoted by the random variable $K_{n}\in\mathbb{R}^{d_{k}}$ and the $(E,\mathcal{E})$-valued sequence $\left\{  X_{n}\right\}_{n\mathbb{\geq}1}$ with $E=R \times H$.  At any time point we define the density $M_\theta\left(\omega, x,{d}x\right)$ to be the transition from $x$ to ${d}x$.  Let $\mathbf{M}^{m_1,m_2}_\theta:\Omega\times E^{m_1} \times \mathcal{E}^{\otimes m_2} \rightarrow [0,1]$ be the transition density of an SMC algorithm (see Section 4) that one can use to simulate $m_2$ samples of ${d}x$ conditioned upon $m_1$ samples of $x$:
\begin{align}\label{eq:Mm1m2defnsgtues}
 \mathbf{M}^{m_1,m_2}_\theta\left(\omega, x,{d}x\right)&=\prod_{i=1}^{m_2} \Phi_\theta^{\omega,m_1}\left( \eta_{\omega}^{m_1} \right) \left({d}x^i\right) \\ \nonumber
\Phi_\theta^{\omega,m_1}\left( \eta_{\omega}^{m_1} \right) \left({d}x^i\right)&=
\frac{\frac{1}{m_1}\sum_{j=1}^{m_1} W\left(\omega,x^j\right)M_\theta\left(\omega, x^j,{d}x^i\right)}{\frac{1}{m_1}\sum_{j=1}^{m_1} W\left(\omega,x^j\right)} \\ \nonumber
 \eta_{\omega}^{m_1} \left( W \right) &= \frac{1}{m_1}\sum_{j=1}^{m_1} W\left(\omega,x^j\right), \nonumber
\end{align}
where $W\left(\omega,x^j\right)$ is a weight assigned to a sample of $x$.  Also, define $\mathbf{W}^{m_1}\left(\omega,x\right) = \frac{1}{m_1} \sum_{j=1}^{m_1} W\left(\omega,x^j\right)$, and define the additive functional
$$
\mathbf{f}^{m_1}\left(\omega,x\right) = \frac{1}{m_1} \sum_{j=1}^{m_1} f\left(\omega,x^j\right),
$$
where $f:\Omega\times E \rightarrow (0,\infty)$.  To make our work easier to follow, we further establish some kernel and operator notation in Table \ref{table:kernelnotation} in Section B of the appendix.

Let $\widetilde{\mathbf{M}}^{m_1,m_2}_\theta:\Omega\times E^{m_1} \times \mathcal{E}^{\otimes m_2} \rightarrow [0,1]$ be some other SMC transition density, which may or may not be the same as $\mathbf{M}^{m_1,m_2}_\theta$, and define a family of kernels $\mathbb{M}^{m_1,m_2}$ similar to as in \cite{Whiteley_2013}:
\begin{definition}\label{def:coolM}
Any $\widetilde{\mathbf{M}}^{m_1,m_2}_\theta$ is said to be a member of $\mathbb{M}^{m_1,m_2}$ if and only if there exist positive, finite constants $(\widetilde{\epsilon}_{-},\widetilde{\epsilon}_{+})$ and probability measures $\nu\in\mathscr{P}(E)$ and $\widetilde{\nu}\in\mathscr{P}(E^{m_2})$ such that
\begin{enumerate}
 \item{$ \widetilde{\nu}\left(\cdot\right)\widetilde{\epsilon}_{-} \leq \widetilde{\mathbf{M}}^{m_1,m_2}_\theta\left(\omega, x,\cdot\right) \leq \widetilde{\epsilon}_{+}\widetilde{\nu}\left(\cdot\right) \quad \forall \left(\omega, x\right)\in\Omega\times E^{m_1} $,}
 \item{$ \nu^{\otimes m_2} $ is dominated by $ \widetilde{\nu} $, and}
 \item{$ \int_{E^{m_2}} \left( \frac{\mathrm{d}\nu^{\otimes m_2}}{\mathrm{d}\widetilde{\nu}} \left(x^{'}\right)\right)^2 \widetilde{\nu}\left(dx^{'}\right)<\infty $.}
\end{enumerate}
Furthermore, when $\widetilde{\mathbf{M}}^{m_1,m_2}_\theta$ is a member of $\mathbb{M}^{m_1,m_2}$, we write 
\begin{equation}\label{eq:partofcoolMM}
\phi_\theta^{\omega,m_1,m_2}\left(x,dx\right)=\frac{\mathrm{d}\mathbf{M}^{m_1,m_2}_\theta\left(\omega, x,\cdot \right)}{\mathrm{d}\widetilde{\mathbf{M}}^{m_1,m_2}_\theta\left(\omega, x,\cdot \right)}\left(dx\right),
\end{equation}
which allows us to define the following:
\begin{align*}
&\mathbf{R}^{m_1,m_2}_\theta\left(\omega, x,{d}x^{'}\right) = \mathbf{W}^{m_1}\left(\omega,x\right)^2 \phi_\theta^{\omega,m_1,m_2}\left(x,x^{'}\right)^2 \widetilde{\mathbf{M}}^{m_1,m_2}_\theta\left(\omega, x,{d}x^{'}\right) \\
&\mathbf{J}^{m_1,m_2}_\theta\left(\omega, x\right) = \int_{E^{m_2}} \mathbf{W}^{m_1}\left(\omega,x\right)^2 \phi_\theta^{\omega,m_1,m_2}\left(x,x^{'}\right)^2 \widetilde{\mathbf{M}}^{m_1,m_2}_\theta\left(\omega, x,{d}x^{'}\right)\\
&\mathbf{L}^{m_1,m_2}_\theta\left(\omega, x,{d}x^{'}\right)=\frac{\mathbf{R}^{m_1,m_2}_\theta\left(\omega, x,{d}x^{'}\right)}{\mathbf{J}^{m_1,m_2}_\theta\left(\omega, x\right)}.
\end{align*}
\end{definition}

Finally, this paper frequently refers to the ratio
\begin{align}\label{eq:Wewillfrequentlyrefertotheratiosgthai}
 \widetilde{\mathcal{V}}_{\theta,n}^\omega &=\frac
 {\int_{E^{m_2}} \prod_{k=1}^{n} \mathbf{R}^{m_k,m_{k+1}}_\theta\left(z^k\omega, x,{d}x^{'}\right)}
 {\bigg[ \int_{E} \prod_{k=1}^{n} W\left(z^k\omega,x\right) M_\theta\left(z^k\omega, x,{d}u\right) \bigg]^2}
\end{align}
and to the additive functional
\begin{equation}\label{eq:labeltheaddfunclh}
\mathbf{h}^{m_2}\left(\omega,x\right) = \frac{1}{m_2} \sum_{j=1}^{m_2} h\left(\omega,x^j\right),
\end{equation}
where $h:\Omega\times E \rightarrow (0,\infty)$.
\par

\setcounter{chapter}{3}
\setcounter{section}{3}
\setcounter{subsection}{0}
\setcounter{figure}{0}
\setcounter{equation}{0} 
\noindent {\bf 3. Hidden Markov models}

Allow an $(R,\mathcal{R})$-valued process with states denoted by the random variable $K_{n}\in\mathbb{R}^{d_{k}}$ to be a Markov process.  Assume that we cannot directly observe each $K_{n}$, but we can only indirectly observe each latent state through the random variable $Y_{n}$ (whose properties were defined in Section 2).  Assume also that the observations are statistically independent of one another, conditioned upon the latent process.  This model is called a hidden Markov model \cite{Cappe_2005}, and it can be formally written as
\begin{align}\label{eqn:standardHMMglenridgermhsr}
K_n \mid \left( K_{1:n-1}= k_{1:n-1}, Y_{1:n-1}=y_{1:n-1} \right) &\sim f_\theta \left(\cdot \mid k_{n-1}\right) \\
Y_n \mid \left( K_{1:n}=k_{1:n}, Y_{1:n-1}=y_{1:n-1} \right) &\sim g_\theta \left(\cdot \mid k_n\right), \nonumber
\end{align}
for $n \geq 1$ where $K_1 \sim f_\theta (\cdot \mid k_0) = \mu_{\theta}(\cdot)$.  The parameter $\theta\in\Theta\subseteq\mathbb{R}^{d_{\theta}}$ is static, and it may or may not be known.

When $\theta$ is known, inference on the hidden process at time $n$ relies on the joint density
\begin{equation}\label{eq:HMMjointlitrev}
 \pi_\theta \left(k_{1:n} \mid y_{1:n}\right) = \frac{\gamma_\theta \left(k_{1:n} , y_{1:n}\right)}{Z_{\theta,1:n}} = \frac{\prod_{t=1}^n g_\theta \left(y_t \mid k_t\right)f_\theta \left(k_t \mid k_{t-1}\right)}{\int \prod_{t=1}^n g_\theta \left(y_t \mid k_t\right)f_\theta \left(k_t \mid k_{t-1}\right)\mathrm{d}k_{1:n}}.
\end{equation}
The normalising constant $Z_{\theta,1:n}$ is the probability density of the observations given $\theta$ (i.e., $Z_{\theta,1:n}=p_{\theta} \left( Y_{1:n}=y_{1:n} \right)=p_{\theta} \left( y_{1:n} \right)$).  It is often referred to as the marginal likelihood.  If $\theta$ is unknown, then we would be interested in inferring not only the hidden process but also values for $\theta$.  In this case, we assign a prior density $\pi \left( \theta \right)$, and Bayesian inference at time $n$ relies on the joint density
\begin{equation}\label{eq:thepointofpmcmctarget}
 \pi \left(\theta, k_{1:n} \mid y_{1:n}\right) \propto \pi \left( \theta \right) \gamma_\theta \left(k_{1:n} , y_{1:n}\right).
\end{equation}

\subsection{\normalsize{Intractable likelihood densities}}\label{sec:hmmfilteringsmoothinglitrevabcabcabc}
This paper concentrates on a particular subset of HMMs whose likelihood density function  $g_\theta$ is intractable, thereby making exact calculation of $\gamma_\theta$ impossible (or at least difficult).  We also assume that there does not exist a true unbiased estimate of $g_\theta$.  However, following the method described in \cite{Jasra_2012b}, we can use a biased approximation of the likelihood density.  In more detail, consider the approximation of the joint density \eqref{eq:HMMjointlitrev} given in \cite{Dean_2010} and \cite{Jasra_2012b}:
\begin{equation}\label{eq:pithetaepsilon1118londonjuly}
\pi_\theta^\epsilon \left(k_{1:n} \mid y_{1:n}\right) = \frac{\prod_{t=1}^n g_\theta^\epsilon \left(y_t \mid k_t\right)f_\theta \left(k_t \mid k_{t-1}\right)}{\int \prod_{t=1}^n g_\theta^\epsilon \left(y_t \mid k_t\right)f_\theta \left(k_t \mid k_{t-1}\right)\mathrm{d}k_{1:n}},
\end{equation}
where $g_\theta^\epsilon \left(y_t \mid k_t\right)={\int_{B_{\epsilon}\left(y_t\right)} g_\theta \left(u \mid k_t\right)\mathrm{d}u}/{\int_{B_{\epsilon}\left(y_t\right)} \mathrm{d}u}$.  Under strong assumptions, \cite[Theorem 1]{Jasra_2012b} and \cite[Theorem 2]{Jasra_2012b} show that \eqref{eq:pithetaepsilon1118londonjuly} is a consistent approximation of \eqref{eq:HMMjointlitrev} as $\epsilon$ tends to zero.
\par

\setcounter{chapter}{4}
\setcounter{section}{4}
\setcounter{subsection}{0}
\setcounter{figure}{0}
\setcounter{equation}{0} 
\noindent {\bf 4. Brief review of sequential Monte Carlo}

When the HMMs in Section 3 are impossible or difficult to work with analytically, one can resort to numerical techniques to draw from the models and approximate densities such as \eqref{eq:HMMjointlitrev} and \eqref{eq:pithetaepsilon1118londonjuly}.  Sequential Monte Carlo (SMC) methods comprise a popular collection of approximation techniques for HMMs (see \cite{Doucet_2008},\cite{Doucet_2001}, and \cite{Gordon_1993}).  SMC techniques simulate a collection of $N$ samples (or ``particles'') in parallel, sequentially in time and combine importance sampling and resampling to approximate sequences such as $\pi_\theta \left(k_{1} \mid y_{1}\right), \dots,\pi_\theta \left(k_{1:n} \mid y_{1:n}\right)$.  The sequence of probability densities only must be known up to their additive constants.  The bootstrap particle filter, which is an SMC scheme that first appeared in \cite{Gordon_1993} and can be used to target \eqref{eq:HMMjointlitrev},  is presented here as Algorithm \ref{alg:SMCfiltering}.  To obtain unbiased estimates of the unknown normalising constants, one can use the output of Algorithm \ref{alg:SMCfiltering} to compute the following formula \cite{Del_Moral_2004}: ${\widehat{Z}}_{\theta,1:n}  = \prod_{t=0}^{n-1} [\frac{1}{N}\sum_{i=1}^{N} W\left(z^t\omega,k_{t+1}^i\right)]$.

\begin{algorithm}
\begin{itemize}
\item{ Step 1: For $i\in\{1,\dots,N\}$, sample $K_{1}^{i} \sim \mu_{\theta}\left( \cdot \right)$ and compute the un-normalised weight:
\begin{equation*}
W\left(\omega,k_{1}^i\right) = \frac{\mu_{\theta}\left( k_{1}^i \right) g_{\theta}\left( y_1(\omega) \mid k_{1}^i \right)}{\mu_{\theta}\left( k_{1}^i \right)}=g_{\theta}\left( y_1(\omega) \mid k_{1}^i \right).
\end{equation*}
For $i\in\{1,\dots,N\}$, sample $A_1^i\in\{1,\dots,N\}$ from a discrete distribution on $\{1,\dots,N\}$ with $j^\text{th}$ probability proportional to $W\left(\omega,k_{1}^j\right)$.  The sample $\{ a_1^{1:N} \}$ are the indices of the resampled particles. Set all normalised weights equal to $1/N$, and set $n=2$.}
\item{ Step 2: For $i\in\{1,\dots,N\}$, sample $K_{n}^{i} \mid k_{n-1}^{a(i)} \sim f_{\theta}\left( \cdot \mid k_{n-1}^{a(i)} \right)$ and compute the un-normalised weight:
\begin{equation*}
W\left(z^{n-1}\omega,k_{n}^i\right) = \frac{f_{\theta}\left( k_{n}^{i} \mid k_{n-1}^{a(i)} \right) g_{\theta} \left( y_n(\omega) \mid k_{n}^{i} \right)}{f_{\theta}\left( k_{n}^{i} \mid k_{n-1}^{a(i)} \right)}=g_{\theta} \left( y_n(\omega) \mid k_{n}^{i} \right).
\end{equation*}
For $i\in\{1,\dots,N\}$, sample $A_{1:n}^i\in\{1,\dots,N\}$ from a discrete distribution on $\{1,\dots,N\}$ with $j^\text{th}$ probability proportional to $W\left(z^{n-1}\omega,k_{n}^j\right)$.  Set all normalised weights equal to $1/N$, and set $n=n+1$.  Return to the start of Step 2.}
\end{itemize}
\caption{\label{alg:SMCfiltering}Bootstrap particle filter}
\end{algorithm}

\subsection{\normalsize{Twisted particle filters}}\label{sec:Iamwatchingsomenewsrightnowyeah222}
A study in \cite{Whiteley_2013} has led to a better understanding of how one might obtain an ideal SMC algorithm in the following sense:
\begin{equation}\label{eq:222idealresulttwistedfrankie}
 \frac{1}{n} \text{log} \left( \frac{\mathbb{E}[\widehat{Z}_{\theta,1:n}^2]}{Z_{\theta,1:n}^2} \right) \rightarrow 0 \quad \text{as} \quad n \rightarrow \infty, \quad \mathbb{P}-a.s.,
\end{equation}
where the expectation is taken with respect to the joint density of the samples obtained by the algorithm.  The ideal algorithm basically amounts to introducing a change of measure on the bootstrap particle filter.  To explain this notion in more detail, consider the transition density $\mathbf{M}^{N,N}_\theta:\Omega\times R^{N} \times \mathcal{R}^{\otimes N} \rightarrow [0,1]$ of Algorithm \ref{alg:SMCfiltering}.  The authors of \cite{Whiteley_2013} define the additive, non-negative functional $\mathbf{h}^{N}\left(\omega,k\right)$ of the form \eqref{eq:labeltheaddfunclh}, where each $h$ is a particular eigenfunction,
\begin{equation*}
 h\left(\omega,k\right) = \lim_{n \to \infty} \frac{Q_{\theta,n}^\omega\left(1\right)\left(k\right)}{\Phi_{\theta,n}^{z^{-n}\omega} \left( \sigma \right)Q_{\theta,n}^\omega\left(1\right)},
\end{equation*}
such that $\eta^\omega Q_\theta^\omega\left(\cdot\right) = \lambda_\omega \eta^{z\omega} \left(\cdot\right)$, $Q_\theta^\omega\left( h\left(z\omega,\cdot\right) \right)\left(k\right) = \lambda_\omega h\left(\omega,k\right)$, and $\eta^\omega (h\left(\omega,k\right)) = 1$, for the limit
\begin{equation}\label{eq:bigwaspnest1}
 \eta^\omega (A) = \lim_{n \to \infty} \Phi_{\theta,n}^{z^{-n}\omega} \left( \sigma \right) \left( A \right)
\end{equation}
and the $\mathbb{R}^{+}$-valued, $\mathscr{F}$-measurable eigenvalue
\begin{equation}\label{eq:bigwaspnest2}
\lambda:\omega\in\Omega \rightarrow \eta^\omega \left(W^{\omega}\right).
\end{equation}
One can use the additive functional $\mathbf{h}^N$ to change the measure of the entire particle system generated by Algorithm \ref{alg:SMCfiltering} and replace $\mathbf{M}^{N,N}_\theta\left(\omega, k,{d}k\right)$ with
$$
\widetilde{\mathbf{M}}^{N,N}_\theta\left(\omega, k,{d}k\right)\propto\mathbf{M}^{N,N}_\theta\left(\omega, k,{d}k\right)\mathbf{h}^{N}\left(z\omega,{d}k\right)
$$
to obtain the ideal SMC algorithm of \cite{Whiteley_2013} that achieves \eqref{eq:222idealresulttwistedfrankie}.  This algorithm uses a new estimate of the normalising constant,
\begin{align*}
 \widehat{Z}_{\theta,1:n}&=\prod_{t=0}^{n-1} \frac{\frac{1}{N}\sum_{i=1}^{N} Q_{\theta}^{z^{t-1}\omega}\left(h(z^{t}\omega,\cdot)\right)\left(k_{t}^{a(i)}\right)}{\mathbf{h}^{N}\left(z^t\omega,k\right)}\\ &=\prod_{t=0}^{n-1} \bigg[\frac{1}{N}\sum_{i=1}^{N} W\left(z^t\omega,k_{t+1}^i\right)\bigg]
  \phi_\theta^{z^{t-1}\omega,N,N}\left(k_{t},k_{t+1}\right),
\end{align*}
which is clearly unbiased when the expectation is taken with respect to $\widetilde{\mathbf{M}}_\theta^{N,N}$.  As re-weighting Markov transitions using eigenfunctions is typically referred to as ``twisting'', the authors call their algorithm the twisted particle filter (see Algorithm \ref{alg:SMCtwisted}).

\begin{algorithm}
\begin{itemize}
\item{ Step 0: For $i\in\{1,\dots,N\}$, sample $K_{0}^{i}$ from some appropriate initial distribution and set the un-normalised weight: $W\left(z^{-1}\omega,k_{0}^i\right) = 1$.  Set $n=1$.}
\item{ Step 1: Resampling steps:

 - Sample $U$ from the discrete uniform distribution on $\{1,\dots,N\}$.

 - Sample $A_{n-1}^u\in\{1,\dots,N\}$ from a discrete distribution on $\{1,\dots,N\}$ with $j^\text{th}$ probability proportional to $Q_{\theta}^{z^{n-2}\omega}\left(h(z^{n-1}\omega,\cdot)\right)\left(k_{n-1}^{j}\right)$.

 - For $i\in\{1,\dots,N\}$ and $i\neq u$, sample $A_{n-1}^i\in\{1,\dots,N\}$ from a discrete distribution on $\{1,\dots,N\}$ with $j^\text{th}$ probability proportional to $W\left(z^{n-2}\omega,k_{n-1}^j\right)$.}
\item{ Step 2: Sampling steps:

 - Sample $K_{n}^{u} \mid k_{n-1}^{a(u)} \propto f_{\theta}\left( \cdot \mid k_{n-1}^{a(u)} \right)h(z^{n-1}\omega,\cdot)$.

 - For $i\in\{1,\dots,N\}$ and $i\not=u$, sample $K_{n}^{i} \mid k_{n-1}^{a(i)} \sim f_{\theta}\left( \cdot \mid k_{n-1}^{a(i)} \right)$.}
\item{ Step 3: For $i\in\{1,\dots,N\}$, compute the un-normalised weight:
\begin{equation*}
W\left(z^{n-1}\omega,k_{n}^i\right) = g_{\theta} \left( y_n(\omega) \mid k_{n}^{i} \right).
\end{equation*}
Set $n=n+1$ and return to the start of Step 1.}
\end{itemize}
\caption{\label{alg:SMCtwisted}Twisted bootstrap particle filter}
\end{algorithm}

\subsection{\normalsize{Alive particle filters}}\label{sec:jingoisticspeech}
In the algorithms just discussed, it is necessary to calculate the likelihood density $g_\theta$.  When repetitively calculating $g_\theta$ is not feasible, one can instead target \eqref{eq:pithetaepsilon1118londonjuly} and employ SMC algorithms that utilise approximate Bayesian computation (ABC); see \cite{Jasra_2013}, \cite{Jasra_2012b}, and \cite{Martin_2012} for examples.  We focus here on one particular combination of SMC and ABC: the alive particle filter of \cite{Jasra_2013}, which is printed here as Algorithm \ref{alg:SMCalive}.

Consider an $(E,\mathcal{E})$-valued discrete-time sequence $\left\{  X_{n}\right\}_{n\mathbb{\geq}1}$ with $E=R \times H$.  This process clearly has the Markov property, and it provides a framework through which \eqref{eq:pithetaepsilon1118londonjuly} can be calculated and the HMM of Section 3 can be approximated.  Recall the transition densities of $\left\{  X_{n}\right\}_{n\mathbb{\geq}1}$:
\begin{align*}
M_\theta\left(z^{-1}\omega, x_{0},x_1\right) &= \mu \left( x_{1} \right) =f_\theta\left(k_{1} \mid k_{0}\right) g_\theta \left( u_{1}(\omega) \mid k_{1} \right) \\
 \{ M_\theta\left(z^{n-2}\omega, x_{n-1},x_n\right)&=f_\theta\left(k_{n} \mid k_{n-1}\right)g_\theta \left( u_{n}(\omega) \mid k_{n} \right)\}_{n\mathbb{\geq}2}. \nonumber
\end{align*}
The sequence $\left\{  X_{n}\right\}_{n\mathbb{\geq}1}$ propagates with $k_{n}\sim f_\theta \left( \cdot \mid k_{n-1} \right)$ and $u_n(\omega)\sim g_\theta \left( \cdot \mid k_{n} \right)$.  The values taken by this propagating sequence can be assigned weights
$$
\{W\left(z^{n-1}\omega,x_n\right) = \mathbb{I}_{R \times B_{n,\epsilon} \left(y_n(\omega)\right)}\left( x_n \right)\}_{n\geq 1},
$$
with $x_{n}=(k_{n},u_{n}(\omega))$ and $B_{n,\epsilon} \left(y_n(\omega)\right) \in \mathcal{H}$, which take a value of one when $u_n(\omega) \in B_{n,\epsilon} \left(y_n(\omega)\right)$ and zero otherwise.  When a realisation of the sequence has weights that are each equal to one, then that realisation is an approximate draw from the true HMM.

The authors of \cite{Jasra_2013} use the sequence $\left\{  X_{n}\right\}_{n\mathbb{\geq}1}$ to obtain a biased approximation of an SMC algorithm targeting an HMM with transition $f_\theta$ and likelihood density $g_\theta$ when $g_\theta$ is either impossible or undesirable to compute (assuming it is still possible to simulate from the likelihood distribution).  Basically, the particle filter of \cite{Jasra_2013} commences by simulating $k_{n}\sim f_\theta \left( \cdot \mid k_{n-1} \right)$, simulating $u_n(\omega)\sim g_\theta \left( \cdot \mid k_{n} \right)$, and then considering $k_{n}$ to be a draw from the latent process of the HMM only if $u_n(\omega) \in B_{n,\epsilon} \left(y_n(\omega)\right)$.  The total number of samples drawn at a time point of the algorithm is denoted by the random variable $T_\omega$, where
\begin{equation}\label{eq:alivetwistedtomegasgtues}
 T_{\omega}=\inf\left\{p\geq N : \sum_{i=1}^{p} W\left(\omega,x^i\right)\geq N\right\}.
\end{equation}
In other words, sampling continues at each time point until at least $N$ samples of non-zero weight are obtained.  This practice prevents Algorithm \ref{alg:SMCalive} from dying out, but it does introduce a random running time.

The alive particle filter has an upper bound on its error that does not depend on $n$ \cite[Theorem 3.1]{Jasra_2013}, and its associated unbiased estimate of the normalising constant is given by \cite[Proposition 3.1]{Jasra_2013}:
\begin{align}\label{eq:apfsg}
 \widehat{Z}_{\theta,1:n} =\prod_{t=0}^{n-1} \bigg[\frac{1}{T_{z^t\omega}-1}\sum_{i=1}^{T_{z^t\omega}-1} W\left(z^t\omega,x_{t+1}^i\right)\bigg] =\prod_{t=0}^{n-1} \frac{N-1}{T_{z^t\omega}-1}.
\end{align}
Note that in proving these results, the authors of \cite{Jasra_2013} make use of a nuance in Algorithm \ref{alg:SMCalive}: the last sampled particle is deleted at every time step.

\begin{algorithm}
\begin{itemize}
\item{ Step 1: Set $h=1$ and $S=0$.}
\item{ Step 2: Sample $X_{1}^{h} \sim \mu \left(\cdot \right)$ and compute the un-normalised weight:
\begin{equation*}
W\left(\omega,x_1^h\right) = \mathbb{I}_{R \times B_{1,\epsilon} \left(y_1(\omega)\right)}\left( x_1^h \right).
\end{equation*}
Compute $S = \sum_{i=1}^h W\left(\omega,x_1^i\right)$.  If $S<N$, then set $h=h+1$ and return to the beginning of Step 2.  Otherwise, set $T_\omega=h$ and $n=2$.}
\item{ Step 3: Set $h=1$ and $S=0$.}
\item{ Step 4: Sample $A_{1:n-1}^h\in\{1,\dots,T_{z^{n-2}\omega}-1\}$ from a discrete distribution on $\{1,\dots,T_{z^{n-2}\omega}-1\}$ with $j^\text{th}$ probability $W\left(z^{n-2}\omega,x_{n-1}^j\right)$.  Sample $X_{n}^{h} \mid x_{n-1}^{a(h)} \sim M_\theta\left(z^{n-2}\omega, x_{n-1}^{a(h)},\cdot\right)$ and compute the un-normalised weight:
\begin{equation*}
W\left(z^{n-1}\omega,x_n^h\right) = \mathbb{I}_{R \times B_{n,\epsilon} \left(y_n(\omega)\right)}\left( x_n^h \right).
\end{equation*}
Compute $S = \sum_{i=1}^h W\left(z^{n-1}\omega,x_n^i\right)$.  If $S<N$, then set $h=h+1$ and return to the beginning of Step 4.  Otherwise, set $T_{z^{n-1}\omega}=h$ and $n=n+1$ and return to the start of Step 3.}
\end{itemize}
\caption{\label{alg:SMCalive}Alive particle filter}
\end{algorithm}

Finally, we note that the conditional probability of the stopping time $T_{z^{n-1}\omega}=t_{z^{n-1}\omega}$ and the particles $x_{n}^{1:{t_{z^{n-1}\omega}}}$ generated by Algorithm \ref{alg:SMCalive} at time $n$ is
\begin{align}\label{eq:probofalivesmcalgo}
 &\mathbb{P}\left(x_{n}^{1:{t_{z^{n-1}\omega}}}, t_{z^{n-1}\omega} \mid \mathscr{F}_{n-1} \right) = \binom{t_{z^{n-1}\omega}-1}{N-1}\times \\ \nonumber
&\bigg[ \prod_{i=1}^{t_{z^{n-1}\omega}} \frac{\frac{1}{t_{z^{n-2}\omega}-1}\sum_{l=1}^{t_{z^{n-2}\omega}-1} 
W\left(z^{n-2}\omega,x_{n-1}^{a(l)}\right)M_\theta\left(z^{n-2}\omega, x_{n-1}^{a(l)},x_n^i\right)}{
\mathbf{W}^{t_{z^{n-2}\omega}-1}\left(\omega,x\right)
} \bigg], \nonumber
\end{align}
where $\mathscr{F}_{n-1}$ is the filtration generated by the particle system through time $(n-1)$ and we require that $\sum_{i=1}^{t_{z^{n-1}\omega}} W\left(z^{n-1}\omega,x_n^i\right)=N$ and $W\left(z^{n-1}\omega,x_n^{t_{z^{n-1}\omega}}\right)=1$.  This is a result which is used in the original work appearing below.
\par

\setcounter{chapter}{5}
\setcounter{section}{5}
\setcounter{subsection}{0}
\setcounter{figure}{0}
\setcounter{equation}{0} 
\noindent {\bf 5. Alive twisted sequential Monte Carlo}

In an effort to try to reduce the variance of Algorithm \ref{alg:SMCalive}'s estimate of ${Z}_{\theta,1:n}$, this paper introduces a change of measure on the particle system generated by that algorithm (in the same spirit of how Algorithm \ref{alg:SMCtwisted} improved Algorithm \ref{alg:SMCfiltering}).  We can use an additive functional of the form \eqref{eq:labeltheaddfunclh} to introduce the change of measure (similar to as in \cite{Whiteley_2013} and Section \ref{sec:Iamwatchingsomenewsrightnowyeah222}).  The conditional probability \eqref{eq:probofalivesmcalgo} of the alive particle filter then becomes 
\begin{align}\label{eq:mehulheadsunday}
\widetilde{\mathbb{P}}\left(x_{n}^{1:{t_{z^{n-1}\omega}}}, t_{z^{n-1}\omega} \mid \mathscr{F}_{n-1} \right) \propto \mathbb{P}\left(x_{n}^{1:{t_{z^{n-1}\omega}}}, t_{z^{n-1}\omega} \mid \mathscr{F}_{n-1} \right) \times \\ \frac{1}{t_{z^{n-1}\omega}-1} \sum_{i=1}^{t_{z^{n-1}\omega}-1} h\left(z^{n-1}\omega,x_n^i\right). \nonumber
\end{align}
This expression can be normalised by changing the summand to
\begin{align*}
\hbar \left(z^{n-1}\omega,x_n^i\right) &= h\left(z^{n-1}\omega,x_n^i\right)\times\\& \bigg[W\left(z^{n-1}\omega,x_n^i\right)\cdot\mathbb{I}_{\{a:a\geq N\}}\left( t_{z^{n-1}\omega} \right)\cdot\frac{N-1}{t_{z^{n-1}\omega}-1} \\&+ (1-W\left(z^{n-1}\omega,x_n^i\right))\cdot\mathbb{I}_{\{a:a\geq N+1\}}\left( t_{z^{n-1}\omega} \right)\cdot\frac{t_{z^{n-1}\omega} - N}{t_{z^{n-1}\omega}-1}\bigg]
\end{align*}
and dividing the entire R.H.S.~of \eqref{eq:mehulheadsunday} by $\Phi_\theta^{{z^{n-2}\omega},t_{z^{n-2}\omega}-1}\left( \eta_{z^{n-2}\omega}^{t_{z^{n-2}\omega}-1} \right) \left(h\left(z^{n-1}\omega,\cdot\right)\right)$.  One can sample from the normalised version of \eqref{eq:mehulheadsunday} via Algorithm \ref{alg:SMCalivetwistedsunday}, which is known hereafter as the alive twisted particle filter, or alive twisted SMC.  We stress that Algorithm \ref{alg:SMCalivetwistedsunday} is not all that more complicated to implement than Algorithm \ref{alg:SMCalive}.

\begin{algorithm}
In the following, it is assumed that the first observation of the HMM is given the index of one (e.g., $y_{1}=y_1(z^0\omega)$ and $y_{2}=y_1(z^1\omega)$).
\begin{itemize}
 \item{ Step 0: Set $n=0$.}
 \item{ Step 1: Sample the twisted particle $X_{n+1}^{1}$ from the probability
\begin{equation*}
 \frac{\Phi_\theta^{{z^{n-1}\omega},t_{z^{n-1}\omega}-1}\left( \eta_{z^{n-1}\omega}^{t_{z^{n-1}\omega}-1} \right) \left(dx_{n+1}^{1}\right)\hbar \left(z^{n}\omega,x_{n+1}^{1}\right)}{\Phi_\theta^{{z^{n-1}\omega},t_{z^{n-1}\omega}-1}\left( \eta_{z^{n-1}\omega}^{t_{z^{n-1}\omega}-1} \right)\left(h\left(z^{n}\omega,\cdot\right)\right)}
\end{equation*}
 and compute the un-normalised weight:
\begin{equation*}
W\left(z^n\omega,x_{n+1}^1\right)=\mathbb{I}_{R \times B_{n+1,\epsilon} \left(y_1\left(z^n\omega\right)\right)}\left( x_{n+1}^1 \right).
\end{equation*}
 }
 \item{ Step 2: Set $r=2$ and $S=0$.}
 \item{ Step 3: Sample the non-twisted particle $X_{n+1}^{r}$ from the probability
$\Phi_\theta^{{z^{n-1}\omega},t_{z^{n-1}\omega}-1}\left( \eta_{z^{n-1}\omega}^{t_{z^{n-1}\omega}-1} \right) \left(x_{n+1}^{r}\right)$
and compute the un-normalised weight:
\begin{equation*}
W\left(z^n\omega,x_{n+1}^r\right)=\mathbb{I}_{R \times B_{n+1,\epsilon} \left(y_1\left(z^n\omega\right)\right)}\left( x_{n+1}^r \right).
\end{equation*}
Compute $S = \sum_{i=1}^r W\left(z^n\omega,x_{n+1}^i\right)$.  If $S<N$, then set $r=r+1$ and return to the beginning of Step 3.  Otherwise, set $T_{z^n\omega}=r$.}
 \item{ Step 4: Sample $C_{n+1}$ from the discrete uniform distribution on $\{1,\dots,T_{z^n\omega}-1\}$; this is the index of the twisted particle in Step 1.  Set $n=n+1$ and return to the start of Step 1.}
\end{itemize}
\caption{\label{alg:SMCalivetwistedsunday}Alive twisted particle filter}
\end{algorithm}

For a simulated path $\underline{x}_{1:n}$ generated by Algorithm \ref{alg:SMCalivetwistedsunday}, where $T_{z^{t-1}\omega}$ samples of $x_t$ have been obtained, we have
\begin{align}\label{eq:atpfsg}
 \widehat{Z}_{\theta,1:n}&=\prod_{t=0}^{n-1} \frac{N-1}{T_{z^t\omega}-1}\cdot\frac{\sum_{i=1}^{T_{z^t\omega}-1} Q_{\theta}^{z^{t-1}\omega}\left(h(z^{t}\omega,\cdot)\right)\left(x_{t}^{a(i)}\right)}{\sum_{i=1}^{T_{z^t\omega}-1} W\left(z^t\omega,x_{t+1}^i\right) h\left(z^{t}\omega,x_{t+1}^i\right)} \\ \nonumber
 &=\prod_{t=0}^{n-1} \left(\frac{1}{T_{z^t\omega}-1}\sum_{i=1}^{T_{z^t\omega}-1} W\left(z^t\omega,x_{t+1}^i\right)\right)
  \phi_\theta^{z^{t-1}\omega,T_{z^{t-1}\omega}-1,T_{z^{t}\omega}-1}\left(x_{t},x_{t+1}\right).
\end{align}
This estimate is clearly unbiased when the expectation is taken with respect to the transition densities of Algorithm \ref{alg:SMCalivetwistedsunday}:
 \begin{equation*}
 \widetilde{\mathbf{M}}^{T_{\omega}-1,T_{z\omega}-1}_\theta\left(\omega, x,D\right)=\frac{\int_D \mathbf{M}^{T_{\omega}-1,T_{z\omega}-1}_\theta\left(\omega, x,dx^{'}\right)\mathbf{h}^{T_{z\omega}-1}\left(z \omega, x^{'}\right)}{\int_{E^{T_{z\omega}-1}} \mathbf{M}^{T_{\omega}-1,T_{z\omega}-1}_\theta\left(\omega, x,du\right)\mathbf{h}^{T_{z\omega}-1}\left(z \omega, u\right)},
 \end{equation*}
for all $D \in \mathcal{E}^{\otimes (T_{z\omega}-1)}$, and $\widetilde{\mathbf{M}}^{T_{\omega}-1,T_{z\omega}-1}_\theta$ is a member of $\mathbb{M}^{T_{\omega}-1,T_{z\omega}-1}$.  Of course, any generic choice of the function $h$ is not guaranteed to induce a low variance for \eqref{eq:atpfsg}.  We show below in Section 6 that the unique optimal choice of
\begin{equation}\label{eq:stillmehulsheadjinsunmondaylondon}
 h\left(\omega,x\right) = \lim_{n \to \infty} \frac{Q_{\theta,n}^\omega\left(1\right)\left(x\right)}{\Phi_{\theta,n}^{z^{-n}\omega} \left( \sigma \right)Q_{\theta,n}^\omega\left(1\right)}
\end{equation}
leads to the low variance.
\par

\setcounter{chapter}{6}
\setcounter{section}{6}
\setcounter{subsection}{0}
\setcounter{figure}{0}
\setcounter{equation}{0} 
\noindent {\bf 6. Optimal change in measure}

The purpose of introducing the change of measure on the alive particle filter is to reduce the variance of the algorithm's estimate of the normalising constant, $\widehat{Z}_{\theta,1:n}$.  More specifically, we would like to achieve
\begin{equation*}
\frac{1}{n}\log \widetilde{\mathcal{V}}_{\theta,n}^\omega=
 \frac{1}{n} \text{log} \left( \frac{\mathbb{E}[\widehat{Z}_{\theta,1:n}^2]}{Z_{\theta,1:n}^2} \right) \rightarrow \Upsilon\left(\widetilde{\mathbf{M}}_\theta\right)=0 \quad \text{as} \quad n \rightarrow \infty, \quad \mathbb{P}-a.s.,
\end{equation*}
which is similar to expression \eqref{eq:222idealresulttwistedfrankie}.  The non-negative, finite constant $\Upsilon\left(\widetilde{\mathbf{M}}_\theta\right)$ is some limiting value which depends on the transition density of Algorithm \ref{alg:SMCalivetwistedsunday}.

We show below that the optimal choice of $h$ which leads to $\Upsilon\left(\widetilde{\mathbf{M}}_\theta\right)=0$ is \eqref{eq:stillmehulsheadjinsunmondaylondon}, which also happens to be an eigenfunction and the unique solution to the system of equations
\begin{align}\label{eq:ihavetoleavesoonforcoventgarden1558}
 \eta^\omega Q_\theta^\omega\left(\cdot\right) &= \lambda_\omega \eta^{z\omega} \left(\cdot\right), \quad Q_\theta^\omega\left( h\left(z\omega,\cdot\right) \right)\left(x\right) = \lambda_\omega h\left(\omega,x\right), \quad \eta^\omega (h\left(\omega,x\right)) = 1,
\end{align}
for the limit \eqref{eq:bigwaspnest1} and the $\mathbb{R}^{+}$-valued, $\mathscr{F}$-measurable eigenvalue \eqref{eq:bigwaspnest2}.  In other words, we show that the same change of measure utilised in \cite{Whiteley_2013} can also be used to reduce the variance of unbiased estimates of the normalising constant when the likelihood density is not computable.  The original work of \cite{Whiteley_2013} only considers the case where $M_\theta\left(\omega, x,{d}x\right)$ can be evaluated pointwise.

To prove our result, we adopt slightly different assumptions from those of \cite{Whiteley_2013}:
\begin{hypB}
\label{hyp:H1anal}
The shift operator $z$ preserves $\mathbb{P}$ and is ergodic. 
\end{hypB}
\begin{hypB}
\label{hyp:H2anal}
At any time point,
\begin{equation}\label{eq:assumptionb21}
 \sup_{\omega\in\Omega} \sup_{\left(x,u\right)\in E^2} \frac{M_\theta^\omega\left( W \right)\left( x \right)}{M_\theta^\omega\left( W \right)\left( u \right)}\leq\Delta_1,
\end{equation}
for some $\Delta_1\in(0,\infty)$.  Furthermore, there exist positive, finite constants $(\epsilon_{-},\epsilon_{+})$ and a probability measure $\nu\in\mathscr{P}(E)$ such that 
\begin{equation}\label{eq:assumptionb22}
\nu\left(\cdot\right)\epsilon_{-} \leq M_\theta\left(\omega, x,\cdot\right) \leq \epsilon_{+}\nu\left(\cdot\right) \quad \forall \left(\omega, x\right)\in\Omega\times E.
\end{equation}
Given the definition of the incremental weights, note that \eqref{eq:assumptionb21} and \eqref{eq:assumptionb22} imply, for all $\sigma, \sigma_1, \sigma_2\in\mathscr{P}(E)$:
\begin{equation*}
 \sup_{n\geq 1}\sup_{\omega\in\Omega} \sup_{\left(x,u\right)\in E^2} \frac{\sigma_1 Q_{\theta,n}^\omega\left( x \right)}{\sigma_2 Q_{\theta,n}^\omega\left( u \right)}\leq\Delta_2, \quad
 0 < \sup_{n\geq 1} \sup_{\omega\in\Omega} \sigma Q_{\theta,n}^\omega\left( 1 \right) <\infty,
\end{equation*}
\begin{equation*}
 0 < \sup_{n\geq 1}\sup_{\left(x,u\right)\in E^2} \frac{Q_{\theta,n}^\omega\left(1\right)\left(x\right)}{Q_{\theta,n}^\omega\left(1\right)\left(u\right)}<\infty,
\end{equation*}
for some $\Delta_2\in(0,\infty)$.
\end{hypB}
\begin{hypB}
\label{hyp:H3}
We always fix $N$ such that $1<N<\infty$, and for all $\omega\in\Omega$, $\epsilon$ is always set in such a way that $T_{\omega}$ as in \eqref{eq:alivetwistedtomegasgtues} is finite.
\end{hypB}
Essentially, the first assumption means that the process producing the observations is stationary and ergodic \cite{Whiteley_2013}.  The second and third assumptions effectively place upper and lower bounds on the estimate of the normalising constant, and they place a finite restriction on the running time of Algorithm \ref{alg:SMCalivetwistedsunday}.  We acknowledge that these assumptions are strong, but they are typical of those used in the literature.

In the appendix in Section A, the following theorem is proven for when $h$ is defined as in \eqref{eq:stillmehulsheadjinsunmondaylondon}:
\begin{theorem}\label{theo:analoguetoTheorem1NWAL}
Assume (B\ref{hyp:H1anal}), (B\ref{hyp:H2anal}) and (B\ref{hyp:H3}). For each $\widetilde{\mathbf{M}}^{T_{\omega}-1,T_{z\omega}-1}_\theta$ any member of a $\mathbb{M}^{T_{\omega}-1,T_{z\omega}-1}$, the following are equivalent:
\begin{enumerate}
 \item{$\Upsilon\left(\widetilde{\mathbf{M}}_\theta\right)=0$.}
 \item{For $\mathbb{P}$-almost all $\omega\in\Omega$, $\exists A_\omega \in \mathcal{E}^{\otimes (T_{\omega}-1)}$ such that $\nu^{\otimes (T_{\omega}-1)}\left(A_\omega^c\right)=0$ and $\forall x\in A_\omega$,
 $$
 \widetilde{\mathbf{M}}^{T_{\omega}-1,T_{z\omega}-1}_\theta\left(\omega, x,D\right)=\frac{\int_D \mathbf{M}^{T_{\omega}-1,T_{z\omega}-1}_\theta\left(\omega, x,dx^{'}\right)\mathbf{h}^{T_{z\omega}-1}\left(z \omega, x^{'}\right)}{\int_{E^{T_{z\omega}-1}} \mathbf{M}^{T_{\omega}-1,T_{z\omega}-1}_\theta\left(\omega, x,du\right)\mathbf{h}^{T_{z\omega}-1}\left(z \omega, u\right)}$$
 for all $D \in \mathcal{E}^{\otimes (T_{z\omega}-1)}$.
 }
 \item{For $\mathbb{P}$-almost all $\omega\in\Omega$, $\sup_{n} \widetilde{\mathcal{V}}_{\theta,n}^\omega<\infty$.}
\end{enumerate}
\end{theorem}
This theorem (which is analogous to \cite[Theorem 1]{Whiteley_2013}) states that there is a unique choice for the change in measure of the particle system that, when analytically available, leads to $\Upsilon\left(\widetilde{\mathbf{M}}_\theta\right)=0$.  However, that optimal $h$ often needs to be approximated.  In the following section, we implement Algorithm \ref{alg:SMCalivetwistedsunday} on an example where the exact form of $h$ needs to be approximated.  The numerical illustration shows that under certain scenarios, the approximation of $h$ is sufficient to reduce the variance of $\widehat{Z}_{\theta,1:n}$.
\par

\setcounter{chapter}{7}
\setcounter{section}{7}
\setcounter{subsection}{0}
\setcounter{figure}{0}
\setcounter{equation}{0} 
\noindent {\bf 7. Implementation of alive twisted SMC}

We compare the variability of the alive particle filter's \eqref{eq:apfsg} to that of the alive twisted particle filter's \eqref{eq:atpfsg}.
We consider a linear Gaussian HMM similar to that of \cite[Section 4.4]{Whiteley_2013}:
\begin{align}\label{eq:WhiteleyHMM2012}
K_0 &\sim \mathcal{N}\left( 0, \nu^2 \right)\\ \nonumber
K_{n} \mid \left( K_{1:n-1}=k_{1:n-1}, Y_{1:n-1}=y_{1:n-1} \right) &\sim \mathcal{N}\left( 0.9k_{n-1}, \nu^2 \right) = f_\theta \left( k_{n} \mid k_{n-1} \right)\\ \nonumber
Y_{n} \mid \left( K_{1:n}=k_{1:n}, Y_{1:n-1}=y_{1:n-1} \right) &\sim \mathcal{N}\left( k_{n}, \tau^2 \right) = g_\theta \left( y_{n}(\omega) \mid k_{n} \right), \nonumber 
\end{align}
for $1 \leq n \leq T$.  In our numerical illustrations, we assume it is undesirable to repetitively calculate the density $g_\theta$, but it is possible to use an ABC approximation.  We know from \cite{Whiteley_2013} that the best approximation of $h$ appropriate for a twisted bootstrap particle filter targeting this HMM is $h\left(z^{n-1}\omega,k_n\right)=\pi_\theta\left( y_1(z^{n-1+l}\omega) \mid k_n \right)=\pi_\theta\left( y_{n}(z^l\omega) \mid k_n \right)$, where $l=5$ is a lag length.  As this expression is analytically available for \eqref{eq:WhiteleyHMM2012}, we use this $h$ in our simulations.  Furthermore, given this form of $h$, it is possible to obtain the closed form expression $Q_{\theta}^{z^{n-1}\omega}\left(h(z^{n}\omega,\cdot)\right)\left(x_{n}^{a(i)}\right)
 =\mathbb{I}_{R \times B_{n+1,\epsilon} \left(y_1(z^{n}\omega)\right)}\left( \cdot \right) \pi_\theta\left( y_{1}(z^{n+l}\omega) \mid k_{n}^{a(i)} \right)$.

In our analysis, we calculate $\text{log} [ \mathbb{V}[ \widehat{Z}_{\theta,1:T}^{\text{Algo}\ref{alg:SMCalive}} ] ]-\text{log} [ \mathbb{V}[ \widehat{Z}_{\theta,1:T}^{\text{Algo}\ref{alg:SMCalivetwistedsunday}} ] ]$ for different pairs of values of the state noise $\nu$ and the observation noise $\tau$, where the variance is taken with respect to the appropriate algorithm.  We run $300$ simulations per pair $(\nu,\tau)$.  The experiment is repeated under different values of $N$ for a fixed $T=100$.  We fix the radius of each weight, which is defined as $\epsilon={|u_1(z^{n-1}\omega) - y_1(z^{n-1}\omega)|}/{|y_1(z^{n-1}\omega)|}={|u_n(\omega) - y_n(\omega)|}/{|y_n(\omega)|}$.

The output (see Figure \ref{fig:FigSingapore3}) is similar to the results in \cite[Section 4.4]{Whiteley_2013}.  When the noise values are concentrated around $\nu=\tau=1$, twisting the alive particle filter results in a significant increase in the precision of $\widehat{Z}_{\theta,1:T}$.  We see less of an improvement when $\nu$ and $\tau$ increase, as these are cases where the alive particle filter already performs poorly (at least under the settings that we tested).

\begin{figure}[H]
\begin{center}
{$\text{log} \bigg[ \mathbb{V}\bigg[ \widehat{Z}_{\theta,1:T}^{\text{Algo}\ref{alg:SMCalive}} \bigg] \bigg]-\text{log} \bigg[ \mathbb{V}\bigg[ \widehat{Z}_{\theta,1:T}^{\text{Algo}\ref{alg:SMCalivetwistedsunday}} \bigg] \bigg]$}\\
\scalebox{0.31}{\includegraphics[trim = 35mm 89mm 35mm 106mm, clip]{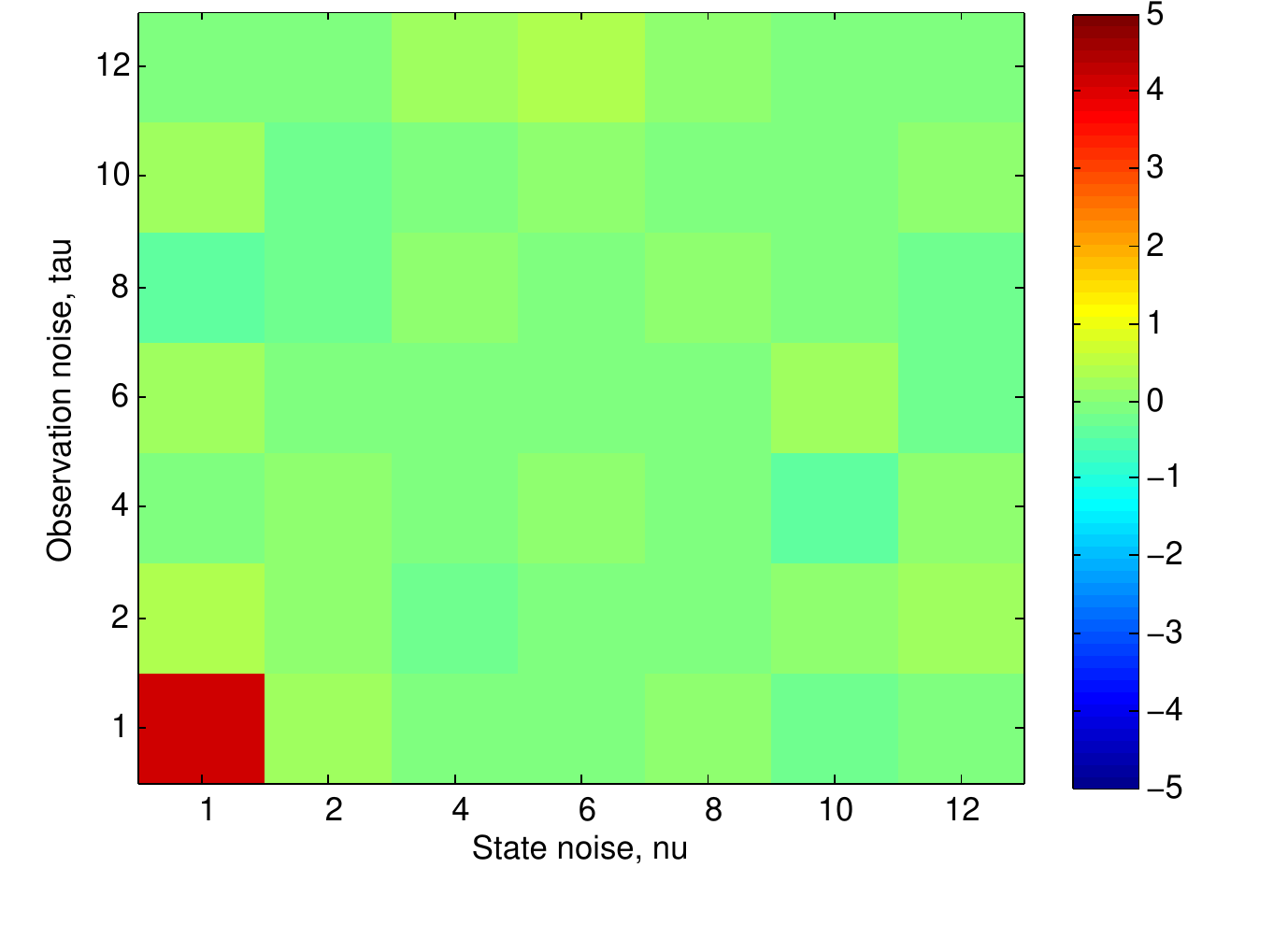}}
\scalebox{0.31}{\includegraphics[trim = 35mm 89mm 35mm 106mm, clip]{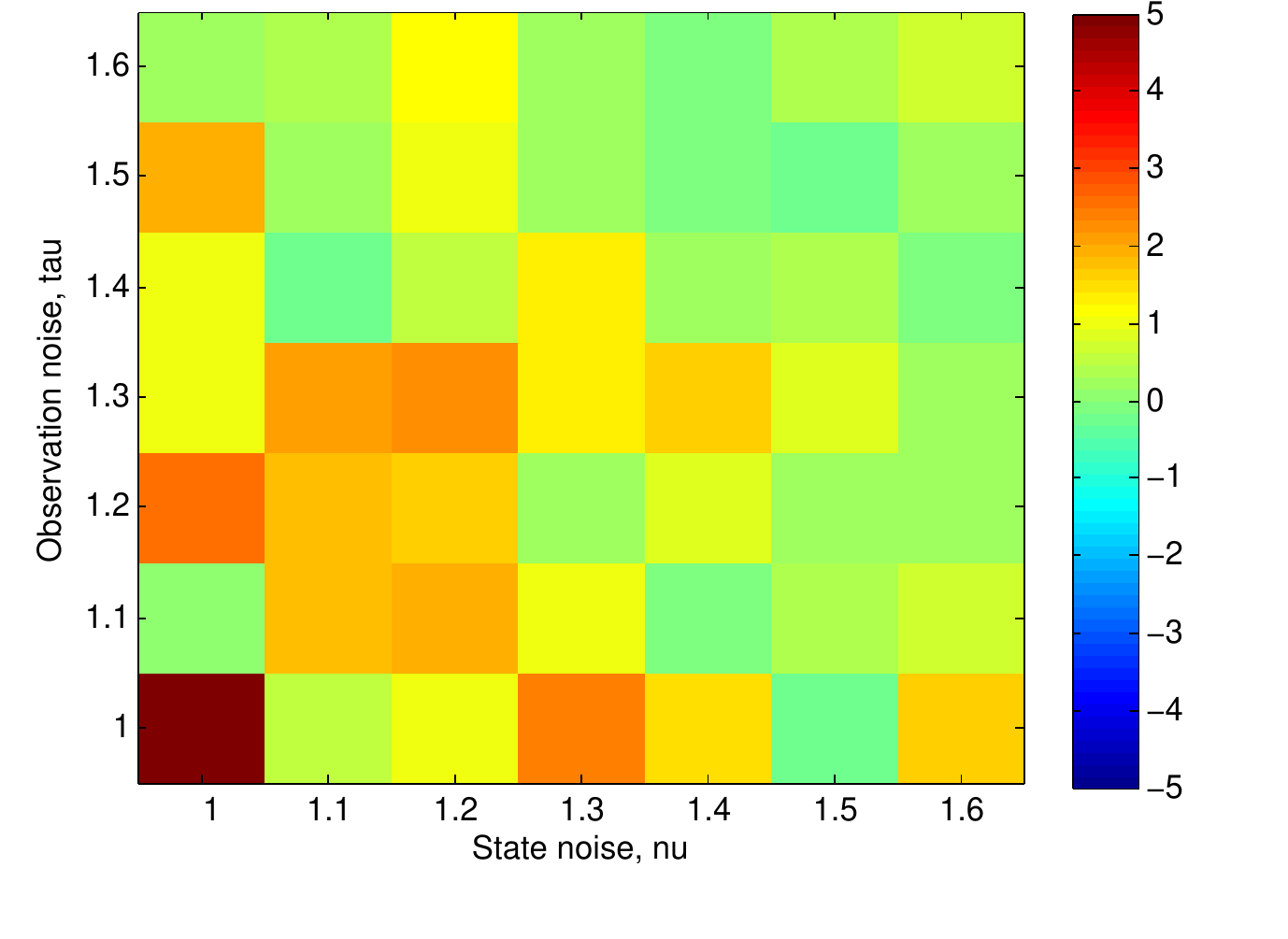}}
\scalebox{0.31}{\includegraphics[trim = 35mm 89mm 35mm 106mm, clip]{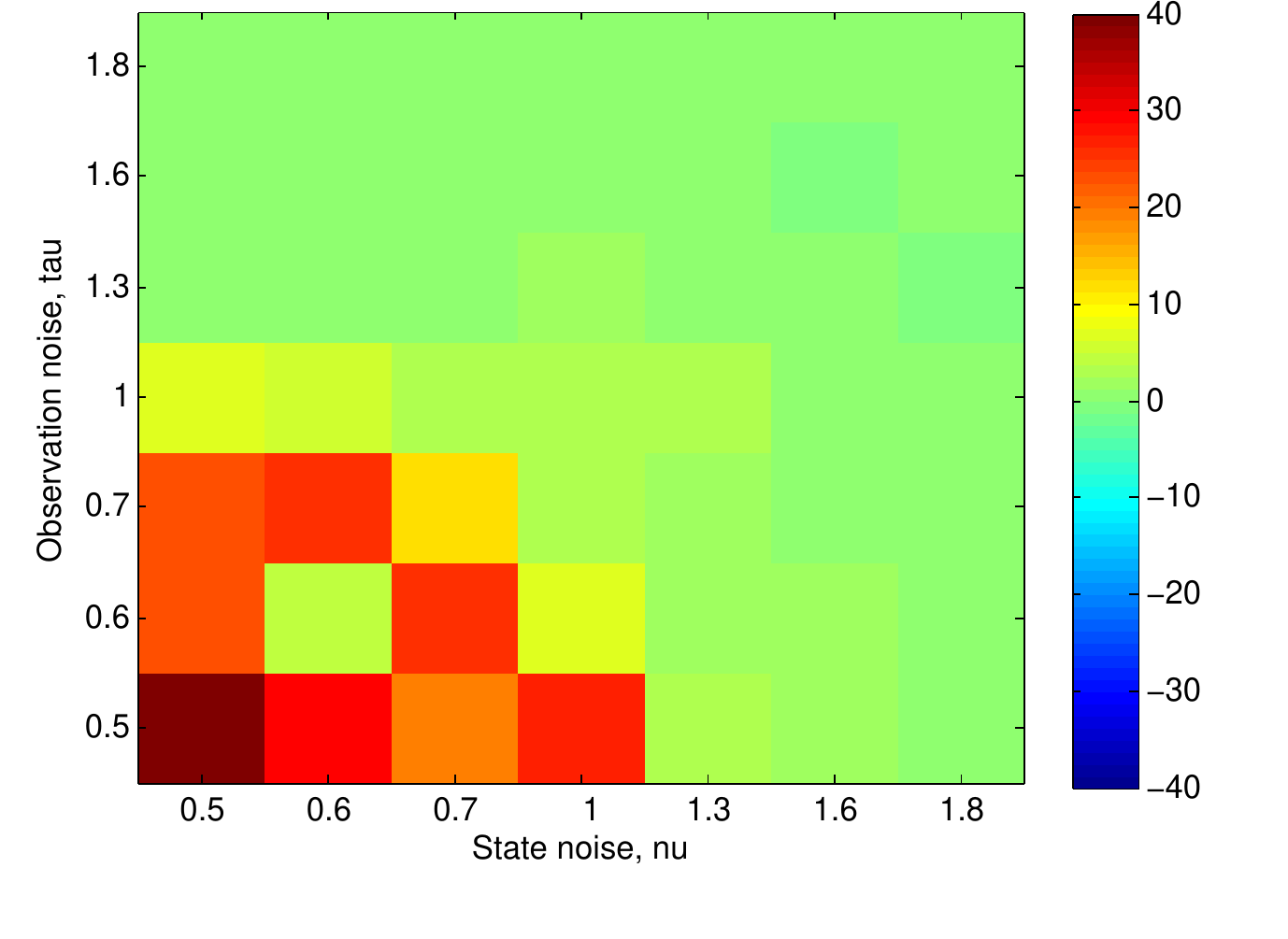}}
\caption[$\quad$Comparing alive SMC to alive twisted SMC]{Each graph measures the difference between the log of the variance of estimates of $Z_{\theta,1:T}$ obtained by the two algorithms.  We set $T=100$ and $\epsilon=1.5$ in each case.  Each block represents $300$ simulations, where $N=1250$ at left and $N=1500$ at centre and right.  Notice that the centre and right graphs are more concentrated around $\nu=\tau=1$ than the graph at left.}
\label{fig:FigSingapore3}
\end{center}
\end{figure}
\par

\setcounter{chapter}{8}
\setcounter{section}{8}
\setcounter{subsection}{0}
\setcounter{figure}{0}
\setcounter{equation}{0} 
\noindent {\bf 8. Alive twisted particle marginal Metropolis-Hastings}

Up to this point, the discussion has focused on sampling from \eqref{eq:pithetaepsilon1118londonjuly}.  Suppose now that the model parameter $\theta$ is unknown, in which case we would sample from
\begin{equation}\label{eq:thepointofpmcmctargetepsilon}
\pi^\epsilon \left(\theta, k_{1:n} \mid y_{1:n}\right) \propto \pi \left( \theta \right) \gamma_\theta^\epsilon \left(k_{1:n} , y_{1:n}\right),
\end{equation}
a density which is similar to \eqref{eq:thepointofpmcmctarget}.  The particle marginal Metropolis-Hastings algorithm \cite{Andrieu_2010} is designed for sampling from densities of the forms \eqref{eq:thepointofpmcmctarget} and \eqref{eq:thepointofpmcmctargetepsilon}.  In a particle marginal Metropolis-Hastings (PMMH) algorithm, one targets an extended density (which is henceforth denoted as $\pi^N$) that yields the true density of interest $\pi$ (or $\pi^\epsilon$) as a marginal; the PMMH employs SMC within a Metropolis-Hastings scheme to sample the variables of the extended target $\pi^N$, and the acceptance ratio of that Metropolis-Hastings scheme is calculated using $\widehat{Z}_{\theta,1:n}$.

The original alive particle filter paper \cite{Jasra_2013} outlines a PMMH algorithm that employs Algorithm \ref{alg:SMCalive}, thereby allowing one to sample from \eqref{eq:thepointofpmcmctargetepsilon}.  However, Section 7 shows that Algorithm \ref{alg:SMCalivetwistedsunday} can indeed outperform Algorithm \ref{alg:SMCalive} in certain scenarios, with the variance of $\widehat{Z}_{\theta,1:n}$ being reduced.  In \cite{Andrieu_2010}, the authors explain that the variance of $\widehat{Z}_{\theta,1:n}$ is critical in the performance of PMMH.  Thus, it is sensible to embed the alive twisted particle filter in PMMH to attempt to expedite the convergence of PMMH.  As $\theta$ is variable, performance improvements in even just some locations of $\Theta$ can be desirable for a PMMH.

It is straightforward to define an alive twisted PMMH as Algorithm \ref{alg:PMMHalivetwistedsunday}, whose extended target density $\pi^N$ is structured as follows.  The joint density of the simulated variables through time $n$ of Algorithm \ref{alg:SMCalivetwistedsunday} is
\begin{align*}
&{\psi}_{\theta}({\underline{x}}_{1:n},{\underline{a}}_{1:n-1},t_{\omega},\dots,t_{z^{n-1}\omega}) \propto \bigg[\binom{t_{\omega}-1}{N-1} \prod_{i=1;\not=c_1}^{t_\omega}  M_\theta\left(z^{-1}\omega, x_{0}^{a(i)},x_1^i\right)\bigg] \times \\ &\bigg[\prod_{j=2}^{n} \binom{t_{z^{j-1}\omega}-1}{N-1} \prod_{i=1;\not=c_j}^{t_{z^{j-1}\omega}} \frac{W\left(z^{j-2}\omega,x_{j-1}^{a(i)}\right)}{\sum_{l=1}^{t_{z^{j-2}\omega}-1}W\left(z^{j-2}\omega,x_{j-1}^{l}\right)} M_\theta\left(z^{j-2}\omega, x_{j-1}^{a(i)},x_j^i\right)\bigg] \times \\
&\bigg[M_\theta\left(z^{-1}\omega, x_{0}^{a(c)},x_1^c\right)h\left(\omega,x_1^c\right)
\bigg] \times \\ &\bigg[\prod_{j=2}^{n} \bigg\{Q_{\theta}^{z^{j-2}\omega}\left(h(z^{j-1}\omega,\cdot)\right)\left(x_{j-1}^{a(c)}\right)\bigg\} M_\theta\left(z^{j-2}\omega, x_{j-1}^{a(c)},x_j^c\right)h\left(z^{j-1}\omega,x_j^c\right)\bigg],
\end{align*}
where we use ${\underline{a}}_{1:n-1}$ to denote the full ancestry of the twisted and non-twisted particles (note also that $c_n$ is being used to denote the index of the twisted particle at any time step $n$).  One can use this expression to establish an extended target as
\begin{align}\label{eq:qqqqqqqqqqqqq}
 &\pi^N \left( d, \theta, {\underline{x}}_{1:n},{\underline{a}}_{1:n-1}, t_{\omega},\dots,t_{z^{n-1}\omega} \mid y_{1:n} \right) \propto \\ \nonumber &\pi\left(\theta\right) {\psi}_{\theta}({\underline{x}}_{1:n},{\underline{a}}_{1:n-1},t_{\omega},\dots,t_{z^{n-1}\omega}) \widehat{Z}_{\theta,1:n} \frac{W\left(z^{n-1}\omega,x_{n}^d\right)}{\sum_{l=1}^{t_{z^{n-1}\omega}-1} W\left(z^{n-1}\omega,x_{n}^l\right)},\nonumber
\end{align}
where $\pi\left(\theta\right)$ is an appropriate prior for the parameter $\theta$.  In both of the above expressions, it is assumed that, at any time step $n$, each sample $\left(\underline{x}_{n}, t_{z^{n-1}\omega}\right)$ satisfies the following: $\sum_{l=1}^{t_{z^{n-1}\omega}-1} W\left(z^{n-1}\omega,x_{n}^l\right)=N-1 \quad \cap \quad W\left(z^{n-1}\omega,x_{n}^{t_{z^{n-1}\omega}}\right)=1$.
Similarly, the proposal density of the PMMH takes the form
\begin{align*}
 &q^N \left( d, \theta, {\underline{x}}_{1:n},{\underline{a}}_{1:n-1},t_{\omega},\dots,t_{z^{n-1}\omega} \right) \propto \\ &q\left( \theta \mid \zeta \right){\psi}_{\theta}({\underline{x}}_{1:n},{\underline{a}}_{1:n-1},t_{\omega},\dots,t_{z^{n-1}\omega})\frac{W\left(z^{n-1}\omega,x_{n}^d\right)}{\sum_{l=1}^{t_{z^{n-1}\omega}-1} W\left(z^{n-1}\omega,x_{n}^l\right)},
\end{align*}
where $q\left( \theta \mid \zeta \right)$ is the density that proposes a new value $\theta\in\Theta$ conditional on a current accepted value $\zeta\in\Theta$.

\begin{algorithm}
\begin{itemize}
 \item{ Step 0: Set $\theta$ arbitrarily.  All remaining random variables can be sampled from their full conditionals defined by the target \eqref{eq:qqqqqqqqqqqqq}:
 
 - Sample ${\underline{x}}_{1:n},{\underline{a}}_{1:n-1},t_{\omega},\dots,t_{z^{n-1}\omega} \mid \cdots$ via Algorithm \ref{alg:SMCalivetwistedsunday} using parameter value $\theta$.

 - Choose $d$ with probability $\frac{W\left(z^{n-1}\omega,x_{n}^d\right)}{\sum_{l=1}^{t_{z^{n-1}\omega}-1} W\left(z^{n-1}\omega,x_{n}^l\right)}$.
 
 Finally, calculate the marginal likelihood estimate, $\widehat{Z}_{\theta,1:n}$, via \eqref{eq:atpfsg}.}
 
 \item{ Step 1: Sample $\theta^* \sim q\left( \cdot \mid \theta \right)$.  All remaining random variables can be sampled from their full conditionals defined by the target \eqref{eq:qqqqqqqqqqqqq}:
 
 - Sample ${\underline{x}}_{1:n}^*,{\underline{a}}_{1:n-1}^*,t_{\omega}^*,\dots,t_{z^{n-1}\omega}^* \mid \cdots$ via Algorithm \ref{alg:SMCalivetwistedsunday} using parameter value $\theta^*$.

 - Choose $d^*$ with probability $\frac{W\left(z^{n-1}\omega,x_{n}^{d^*}\right)}{\sum_{l=1}^{t_{z^{n-1}\omega}-1} W\left(z^{n-1}\omega,x_{n}^l\right)}$.

 Finally, calculate the marginal likelihood estimate, $\widehat{Z}_{\theta^*,1:n}$, via \eqref{eq:atpfsg}.}

 \item{ Step 2: With acceptance probability
 \begin{align*}
  1 &\wedge \frac{\pi(\theta^*)}{\pi(\theta)} \frac{q(\theta \mid \theta^*)}{q(\theta^* \mid \theta)}\frac{\widehat{Z}_{\theta^*,1:n}}{\widehat{Z}_{\theta,1:n}},
 \end{align*}
 set $d=d^*$, $\theta=\theta^*$, ${\underline{x}}_{1:n}={\underline{x}}_{1:n}^*$, ${\underline{a}}_{1:n-1}={\underline{a}}_{1:n-1}^*$, and $t_{\omega},\dots,t_{z^{n-1}\omega}=t_{\omega}^*,\dots,t_{z^{n-1}\omega}^*$.
 
 Return to the beginning of Step 1.}
\end{itemize}
\caption{\label{alg:PMMHalivetwistedsunday}Alive twisted PMMH}
\end{algorithm}
\par

\setcounter{chapter}{9}
\setcounter{section}{9}
\setcounter{subsection}{0}
\setcounter{figure}{0}
\setcounter{equation}{0} 
\noindent {\bf 9. Implementation of alive twisted PMMH}

In the next numerical illustration, we compare the convergence of Algorithm \ref{alg:PMMHalivetwistedsunday} to that of the PMMH employing the non-twisted alive particle filter (see \cite{Jasra_2013}).  We consider a stochastic volatility model which is similar to the one appearing in \cite{Jasra_2013}:
\begin{align*}
K_0 &\sim \mathcal{N}\left( 0, \nu^2 \right)\\ \nonumber
K_{n} \mid \left( K_{1:n-1}=k_{1:n-1}, Y_{1:n-1}=y_{1:n-1} \right) &\sim \mathcal{N}\left( Fk_{n-1}, \nu^2 \right) = f_\theta \left( k_{n} \mid k_{n-1} \right)\\ \nonumber
Y_{n} \mid \left( K_{1:n}=k_{1:n}, Y_{1:n-1}=y_{1:n-1} \right) &\sim \exp{\left(k_n/2\right)}\mathcal{S}\left(\alpha,0.05,\gamma,0\right), \nonumber
\end{align*}
for $1 \leq n \leq T$.  This model is more challenging than the linear Gaussian HMM \eqref{eq:WhiteleyHMM2012} because the probability density functions of the observations are not defined for all parameter values of the stable distribution.  However, the stable distribution is Gaussian when the stability parameter is $\alpha=2$.  Thus, this section uses the same approximation for $h$ that was used in Section 7, and only when calculating 
\begin{equation}\label{eq:frankiegothisboostershotstoday440pm}
h\left(z^{n-1}\omega,x_n\right)=\pi_\theta\left( y_1(z^{n-1+l}\omega) \mid k_n \right)_{\{l=5\}}=\pi_\theta\left( y_{n}(z^l\omega) \mid k_n \right)_{\{l=5\}},
\end{equation}
we assume that the density of the observations is Gaussian.

The observations are daily logarithmic returns of the S\&P 500.  We consider three datasets that each begin with 10th December 2009 and run for $T=200$, $T=500$, or $T=700$ time steps (see Figure \ref{fig:nocountryforoldmen0}).  The datasets are chosen for their different time lengths, to study how $T$ affects the relative performance of the algorithms.

\begin{figure}[H]
\centering
\scalebox{0.45}{\includegraphics[trim = 35mm 90mm 40mm 107mm, clip]{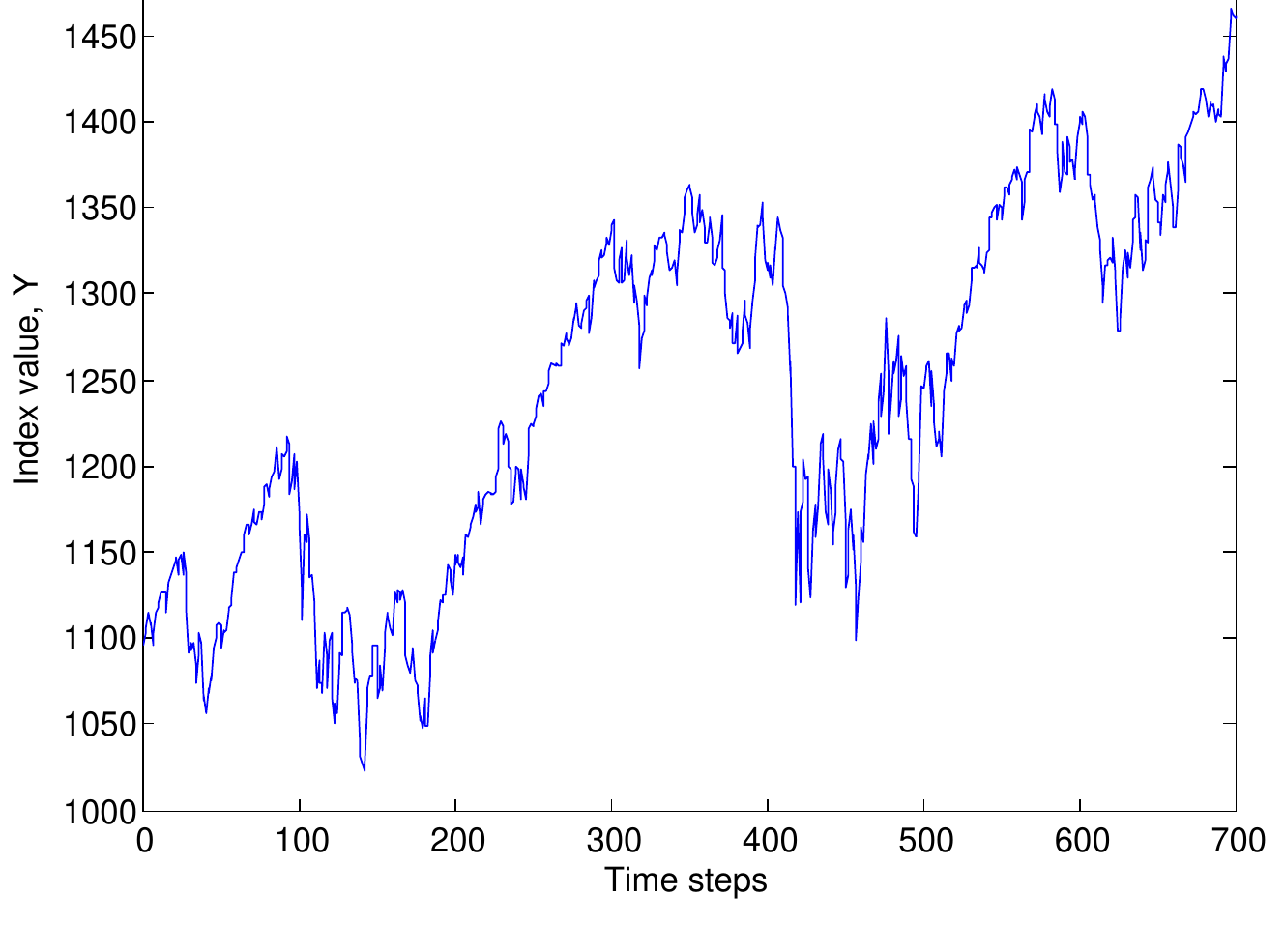}}
\scalebox{0.45}{\includegraphics[trim = 35mm 90mm 40mm 107mm, clip]{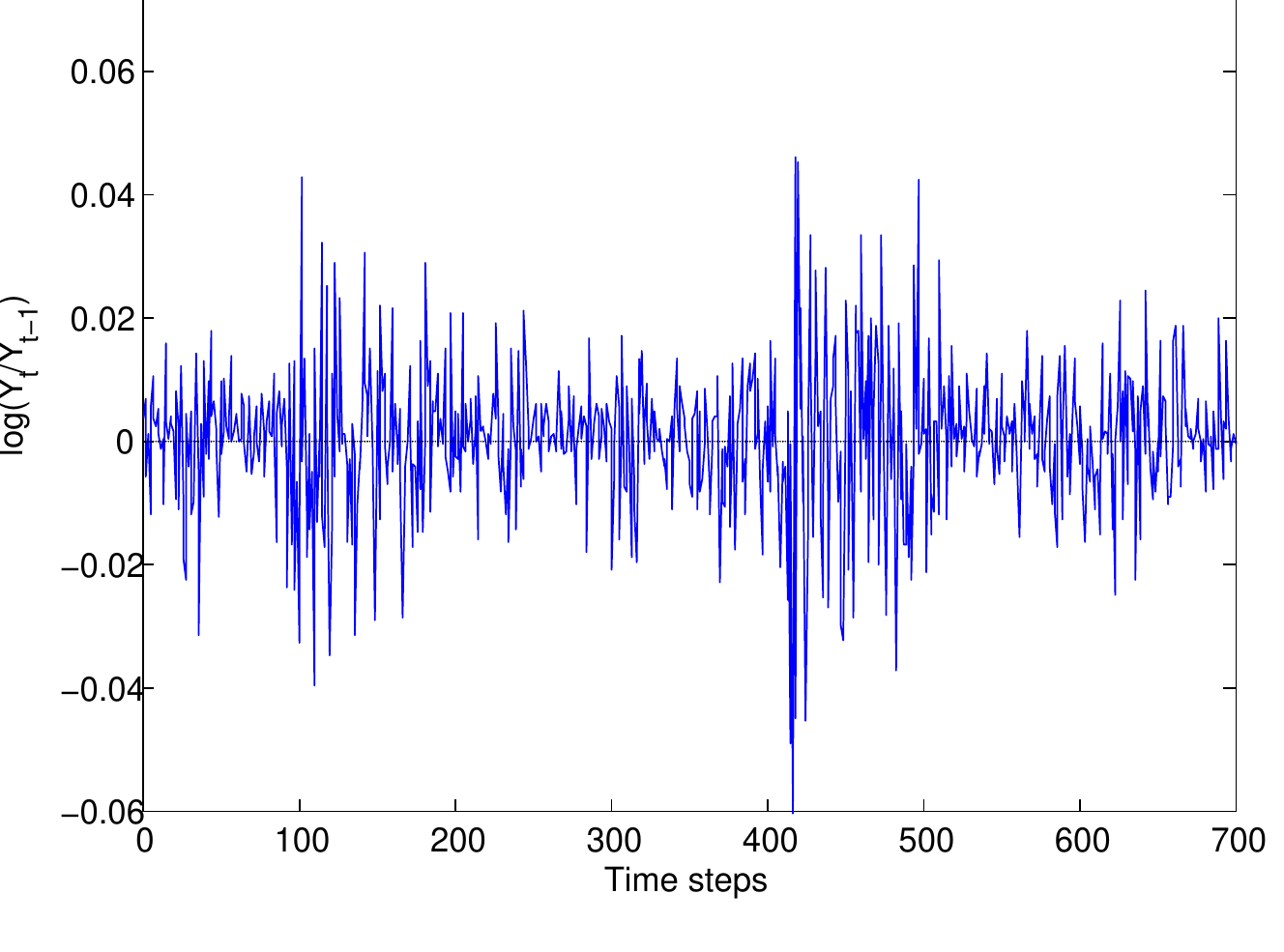}}
\caption[$\quad$Daily closing index value of S\&P 500 and the daily logarithmic returns]{Daily closing index value of S\&P 500 (left) and the daily logarithmic returns (right).  In each plot, the first time step corresponds to 10th December 2009.}
\label{fig:nocountryforoldmen0} 
\end{figure}

Both PMMH algorithms are used to infer the scalars $F\sim\mathcal{N}\left( 0, 0.15 \right)$, ${\nu^{-2}}\sim\mathcal{G}a\left( 2, 100 \right)$, and ${\gamma^{-1}}\sim\mathcal{G}a\left( 2, 1 \right)$.  We try different algorithmic settings for each of the three datasets, choosing $\alpha\in\{1.75,1.95\}$ and $N\in\{50,100,1000\}$; both PMMH schemes have approximately equal running times for equal values of $N$.  Across all datasets, we fix the number of PMMH iterations to $M=100000$, and we fix $\epsilon=3.5$.  Five runs of both algorithms are repeated per group of algorithmic settings.  The proposals for the parameters $\nu^2$ and $\gamma$ are log normal random walks: $\text{log}(\theta^*)=\text{log}(\theta)+\iota$, $\iota\sim\mathcal{N}(0,0.5)$.  The proposal for $F$ is a normal random walk: $\theta^*\sim\mathcal{N}(\theta,1)$.  We track the convergence of the PMMH algorithms using the autocorrelation functions (ACFs) and the trace plots of $F$, $\nu^2$, and $\gamma$.

Both algorithms seem to perform similarly when $\alpha=1.75$, regardless of the values of $T$ or $N$ (results not shown).  This output suggests that \eqref{eq:frankiegothisboostershotstoday440pm} is a poor approximation of the true eigenfunction $h$ in the case where $\alpha=1.75$.  However, when $\alpha=1.95$, the ACF plots (see Figure \ref{fig:nocountryforoldmen7}) show the alive twisted PMMH slightly outperforming the non-twisted alive PMMH.  Thus, it appears \eqref{eq:frankiegothisboostershotstoday440pm} is a fair approximation to the true, optimal $h$ when $\alpha=1.95$.  We only present the output for $T=500$ (see Figures \ref{fig:nocountryforoldmen7} and \ref{fig:nocountryforoldmen8}), as the results are similar for the slightly different values of $T=200$ and $T=700$.

\begin{figure}[H]
\centering
\scalebox{0.32}{\includegraphics[trim = 35mm 90mm 40mm 107mm, clip]{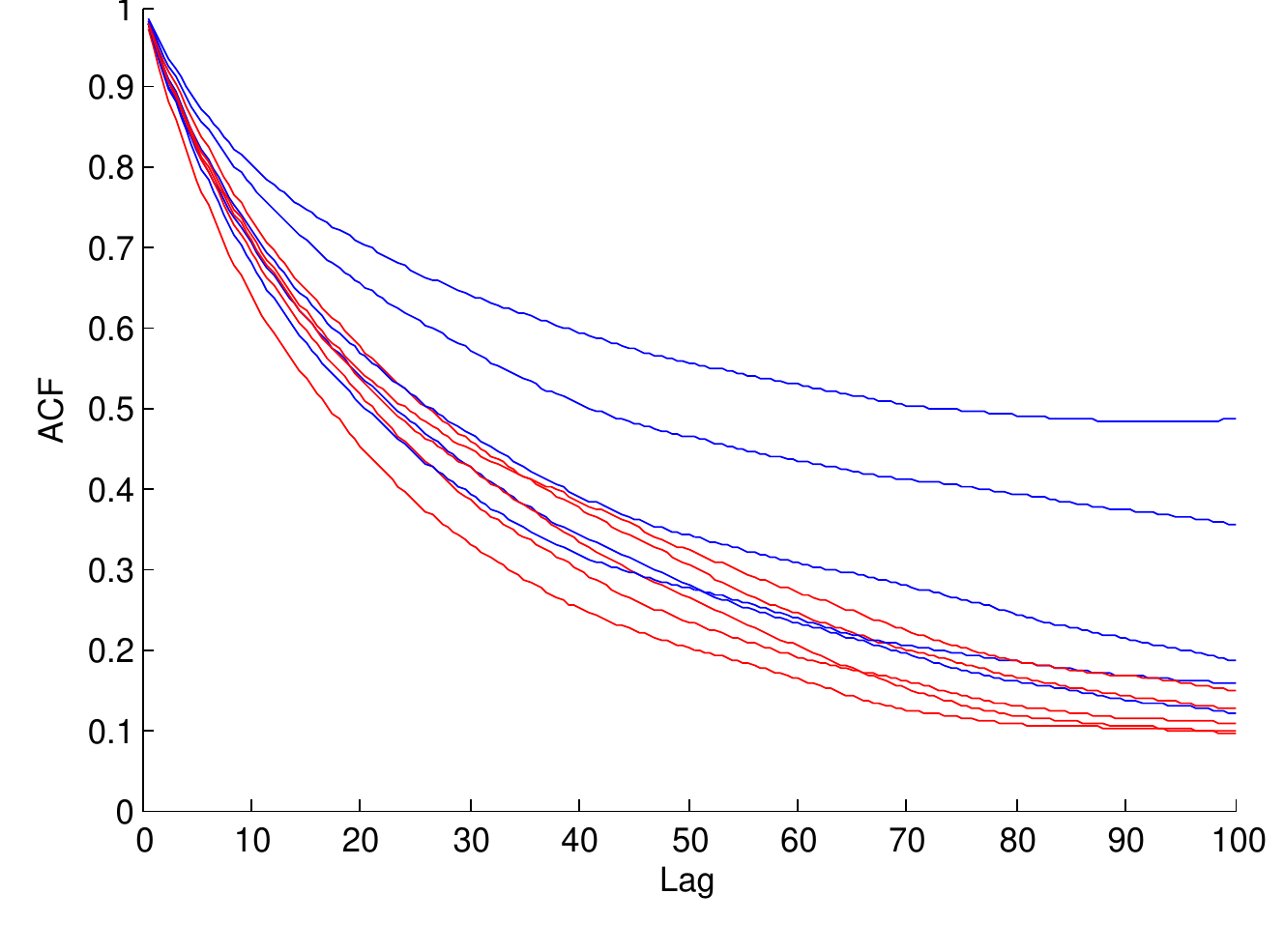}}
\scalebox{0.32}{\includegraphics[trim = 35mm 90mm 40mm 107mm, clip]{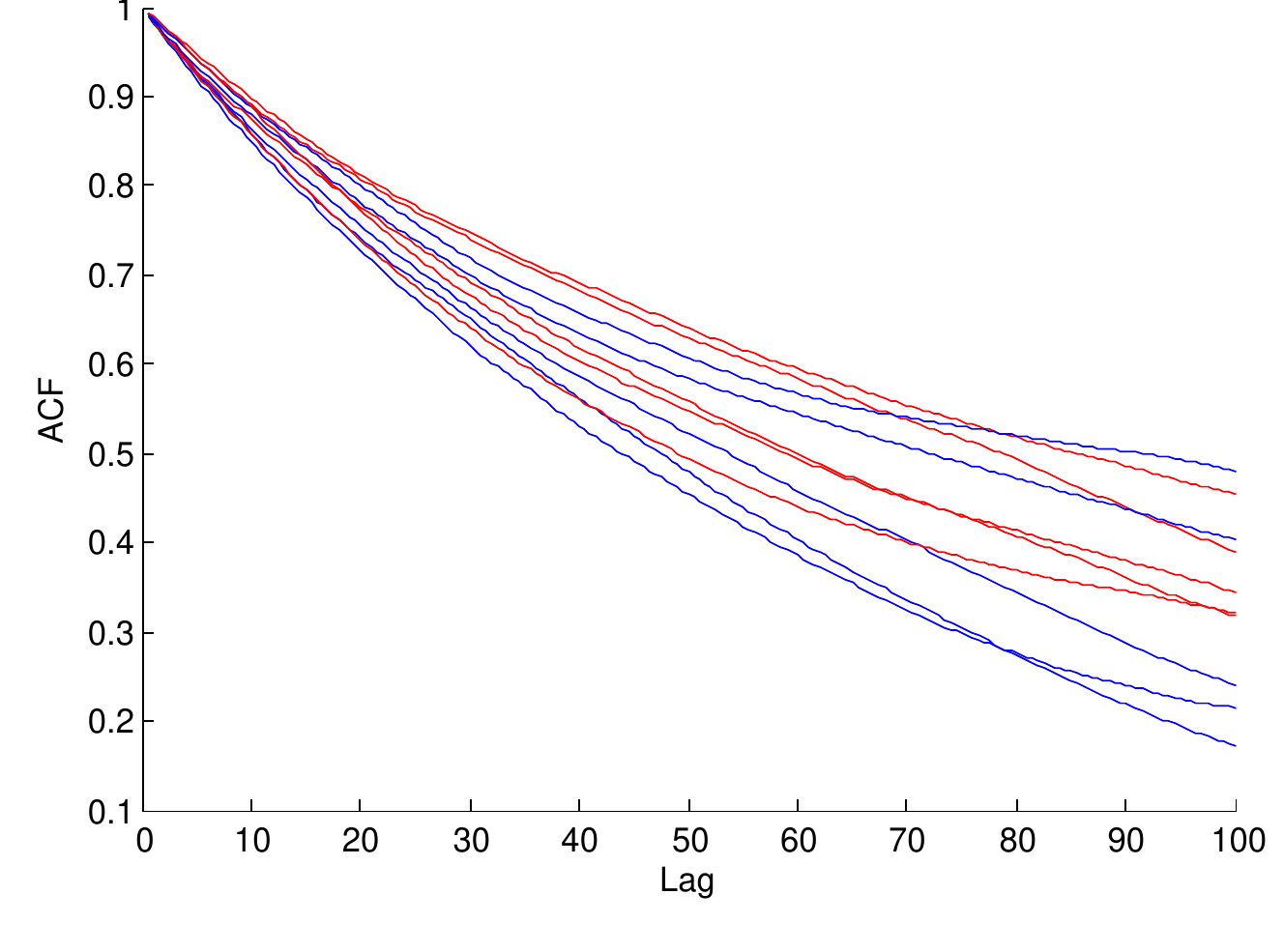}}
\scalebox{0.32}{\includegraphics[trim = 35mm 90mm 40mm 107mm, clip]{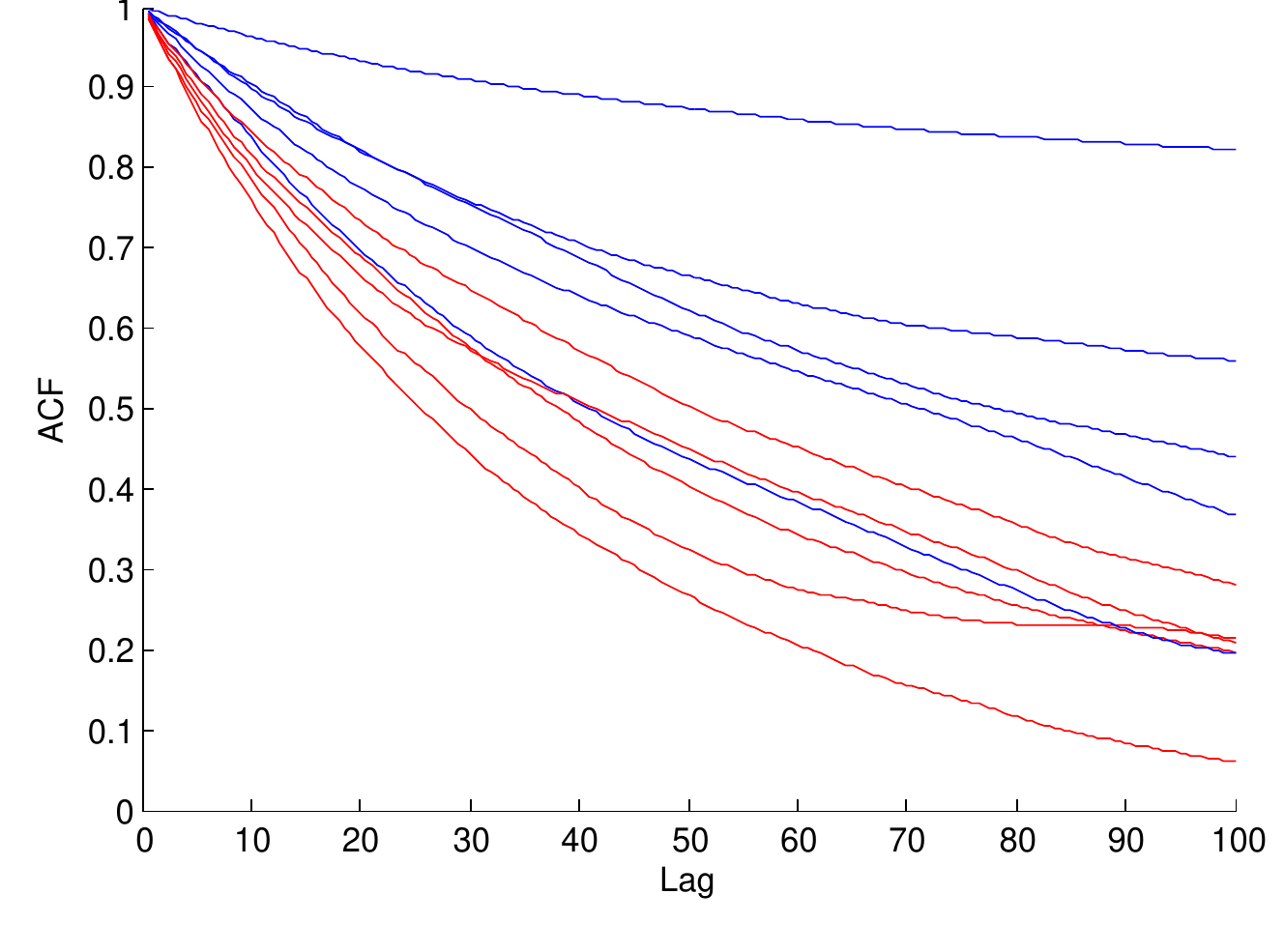}} \\
\scalebox{0.32}{\includegraphics[trim = 35mm 90mm 40mm 107mm, clip]{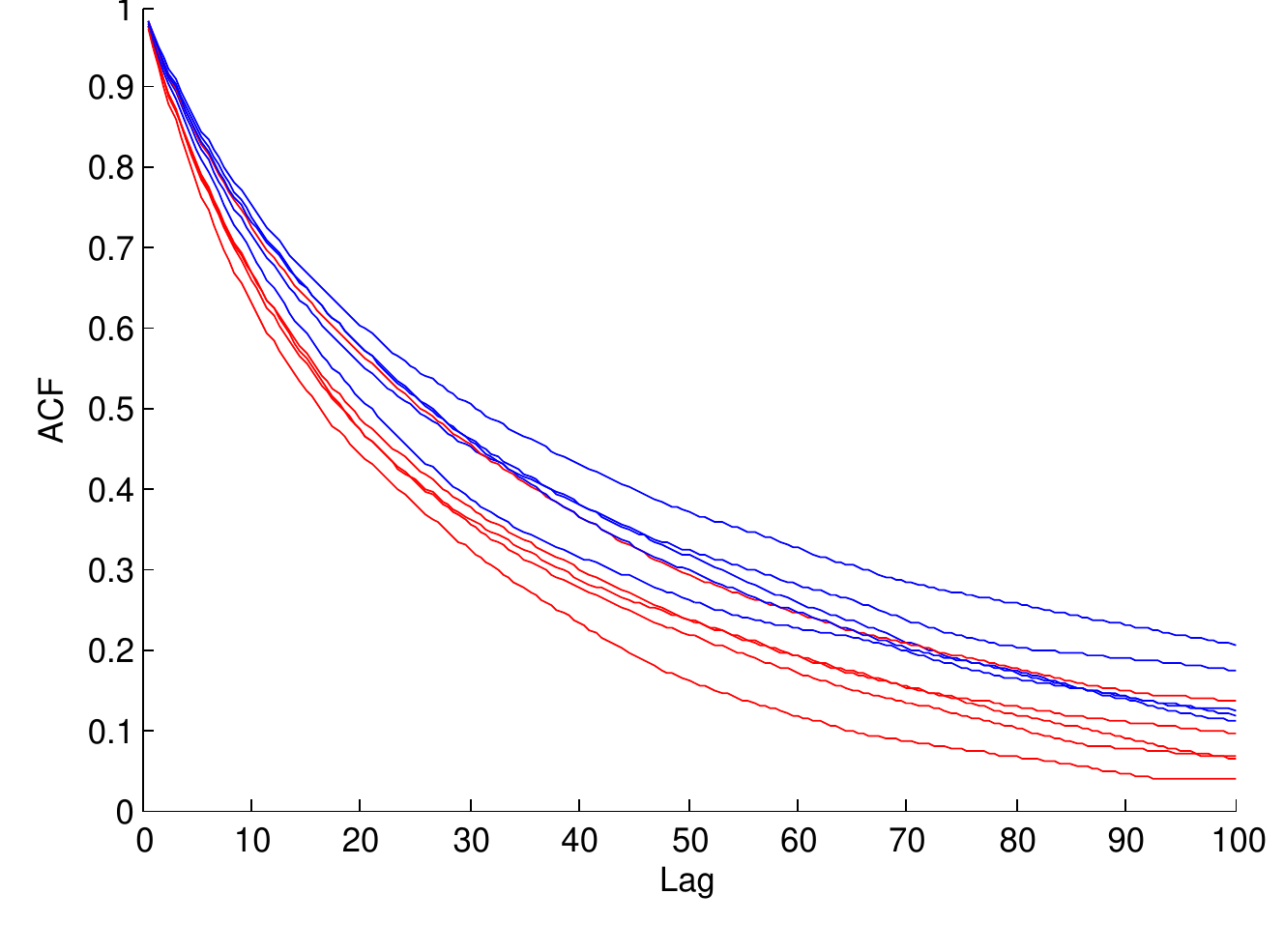}}
\scalebox{0.32}{\includegraphics[trim = 35mm 90mm 40mm 107mm, clip]{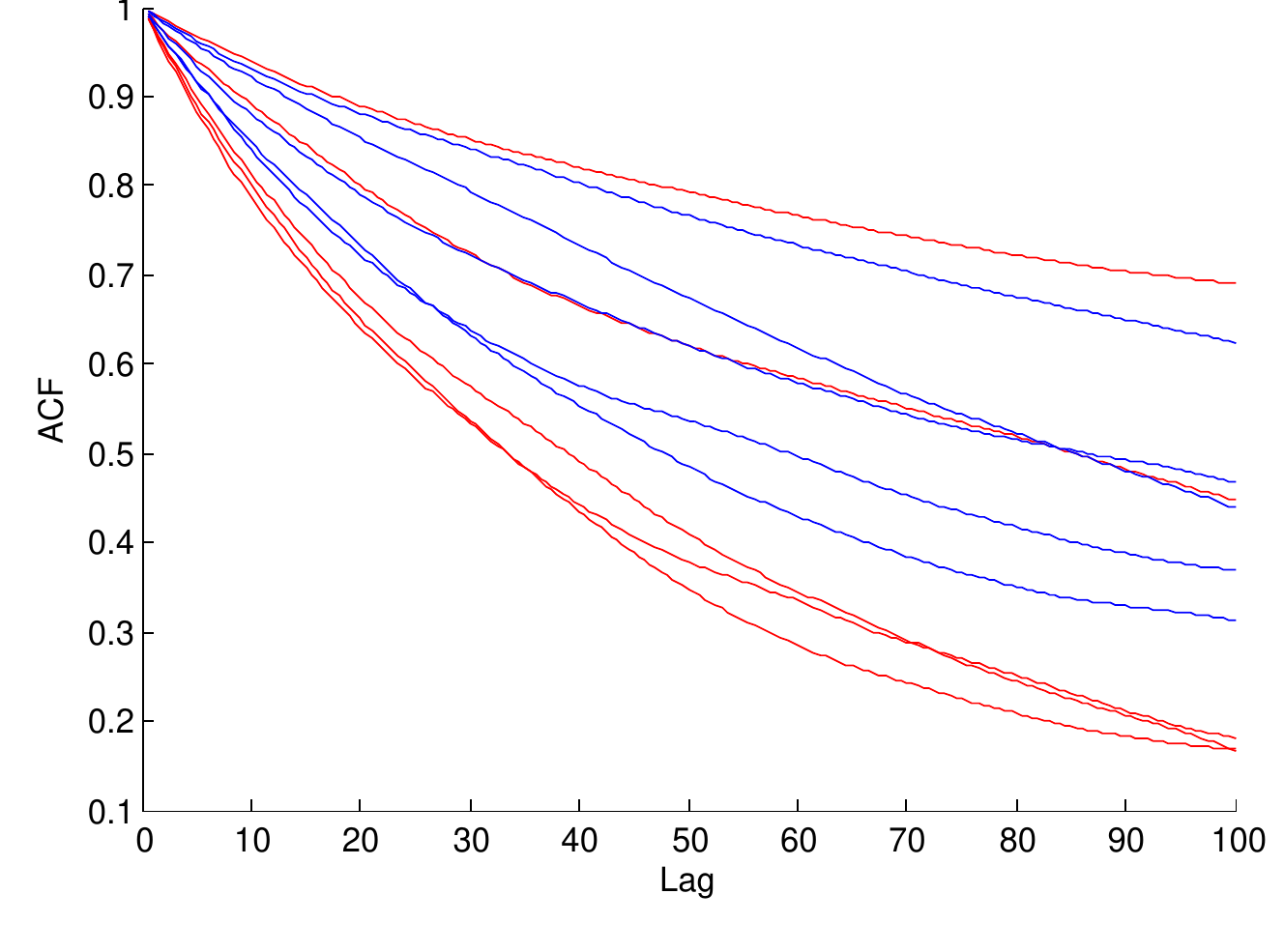}}
\scalebox{0.32}{\includegraphics[trim = 35mm 90mm 40mm 107mm, clip]{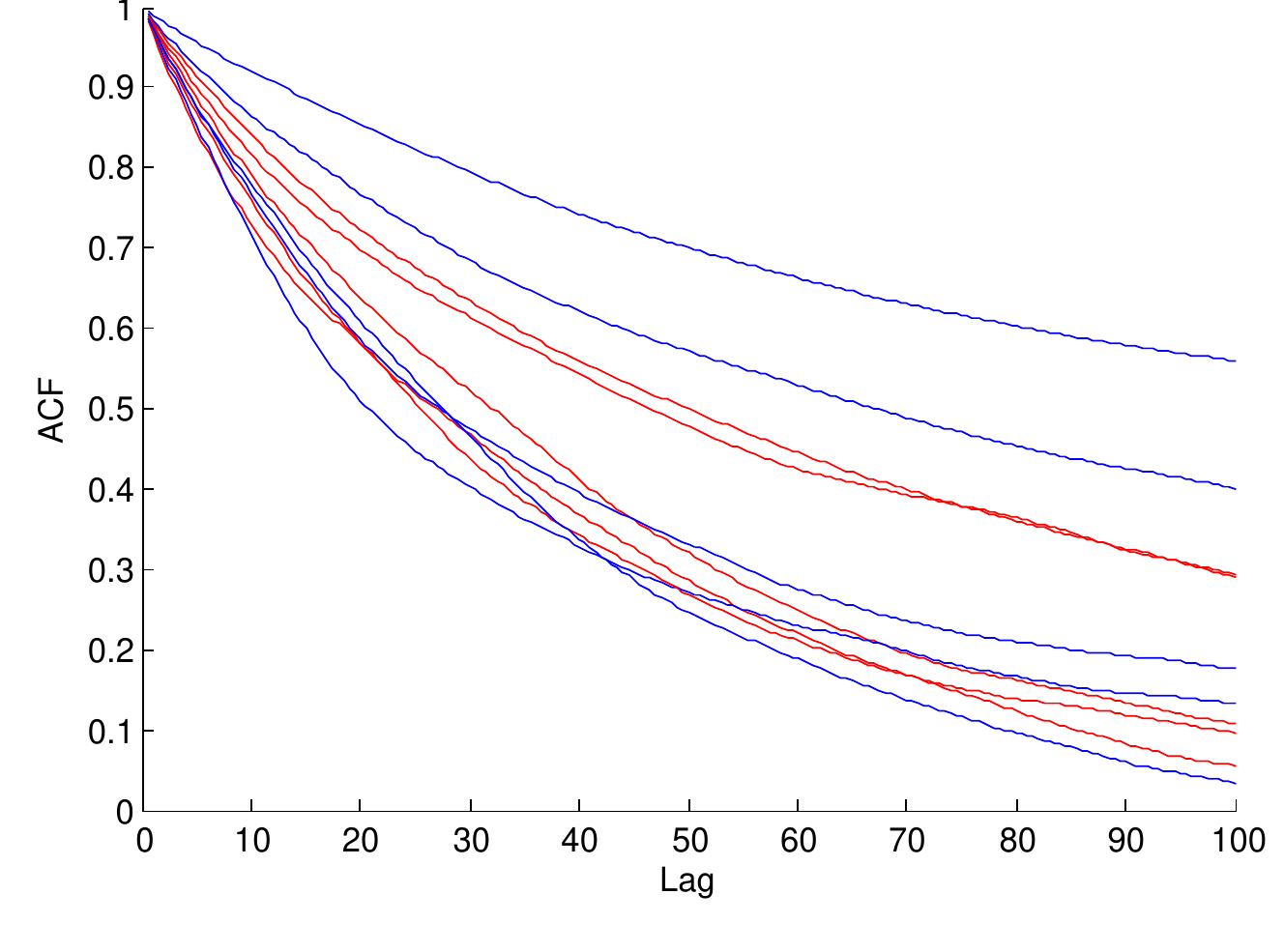}}
\caption[$\quad$Comparing alive PMMH to alive twisted PMMH]{Results for $N=50$ (top) and $N=100$ (bottom).  ACF plots for $F$, $\nu^2$, and $\gamma$ (from left to right).
The alive PMMH corresponds to the blue lines, and the alive twisted PMMH corresponds to the red lines.}
\label{fig:nocountryforoldmen7} 
\end{figure}

\begin{figure}[H]
\centering
\scalebox{0.3}{\includegraphics[trim = 40mm 90mm 40mm 107mm, clip]{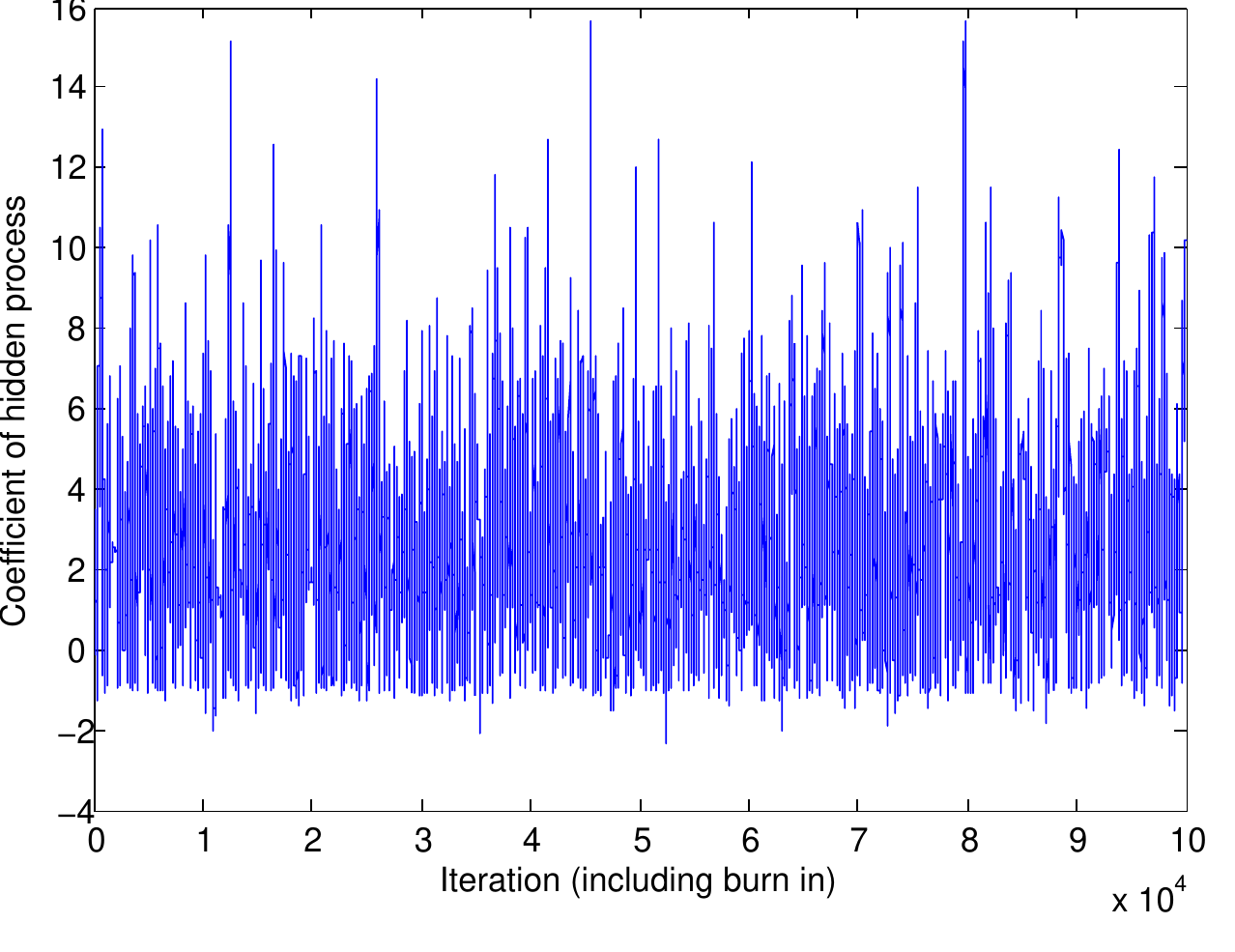}}
\scalebox{0.3}{\includegraphics[trim = 40mm 90mm 40mm 107mm, clip]{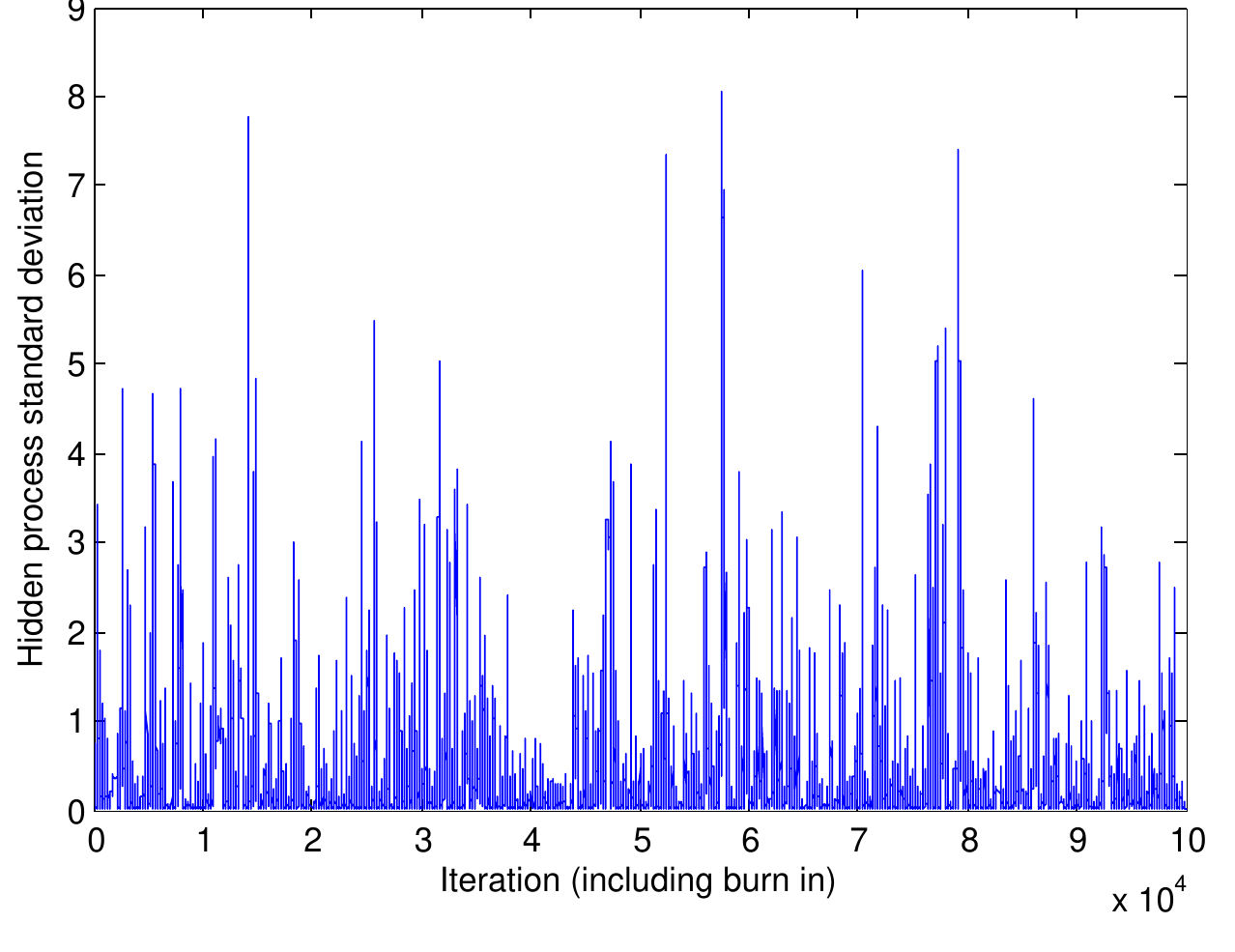}}
\scalebox{0.3}{\includegraphics[trim = 40mm 90mm 40mm 107mm, clip]{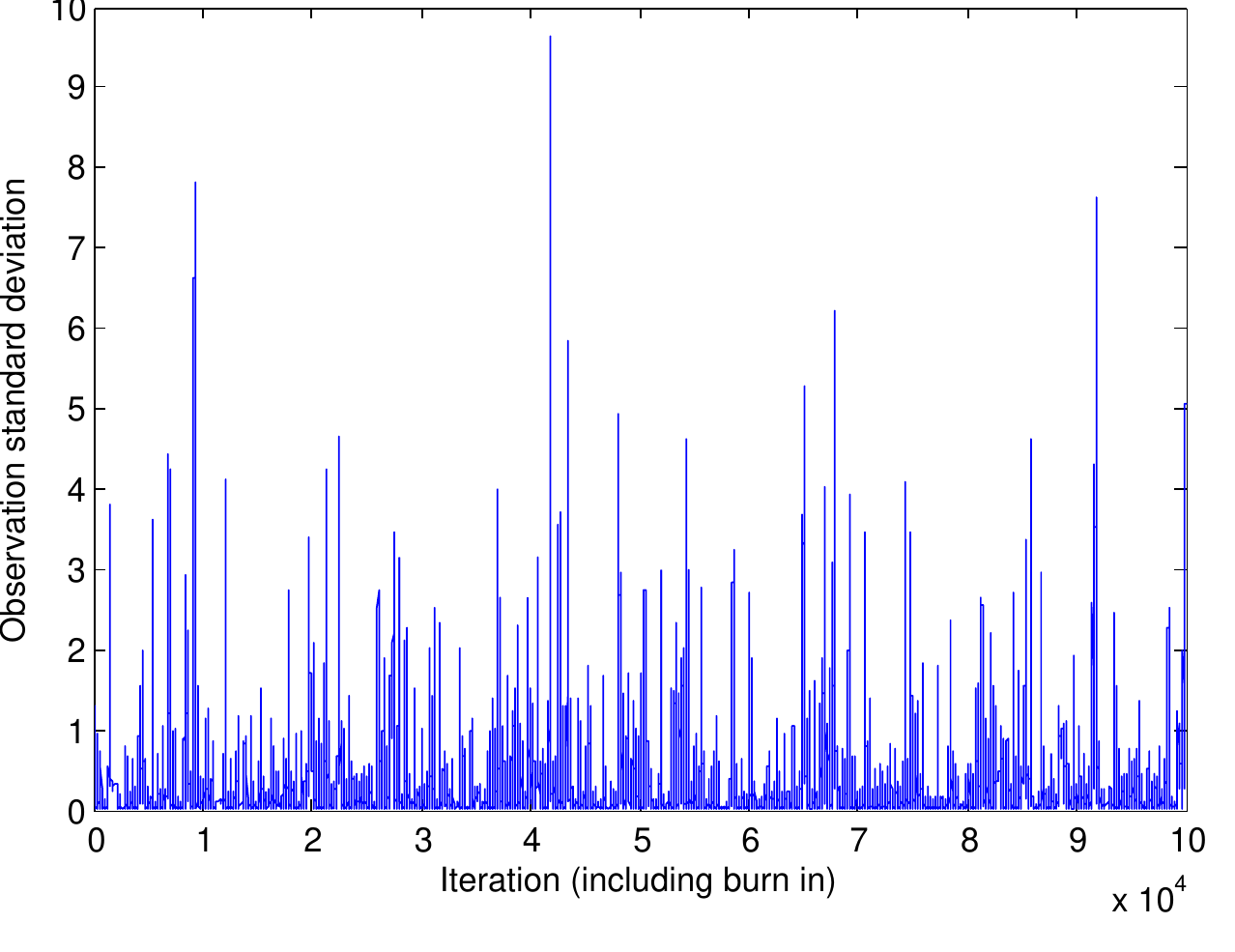}} \\
\scalebox{0.3}{\includegraphics[trim = 40mm 90mm 40mm 107mm, clip]{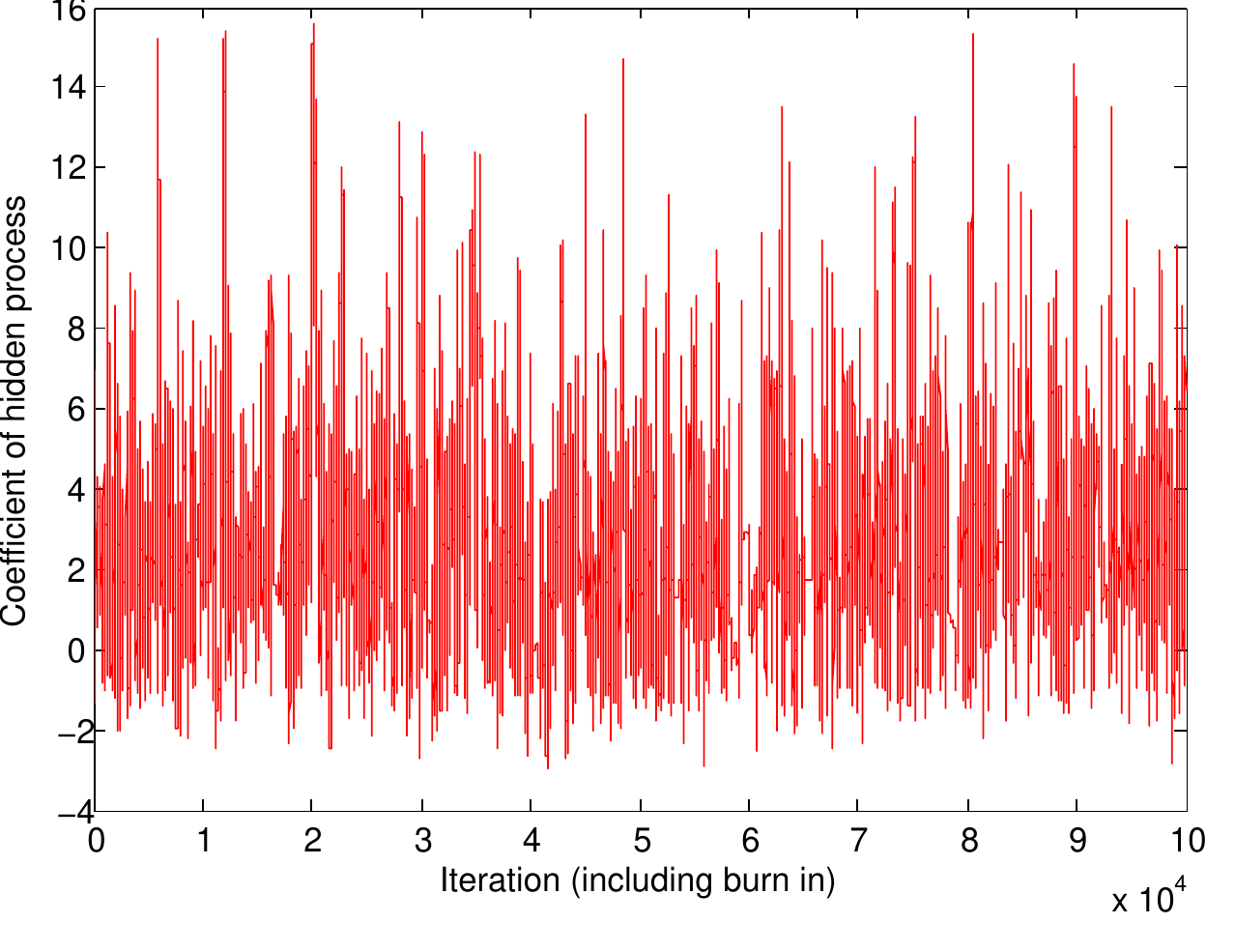}}
\scalebox{0.3}{\includegraphics[trim = 40mm 90mm 40mm 107mm, clip]{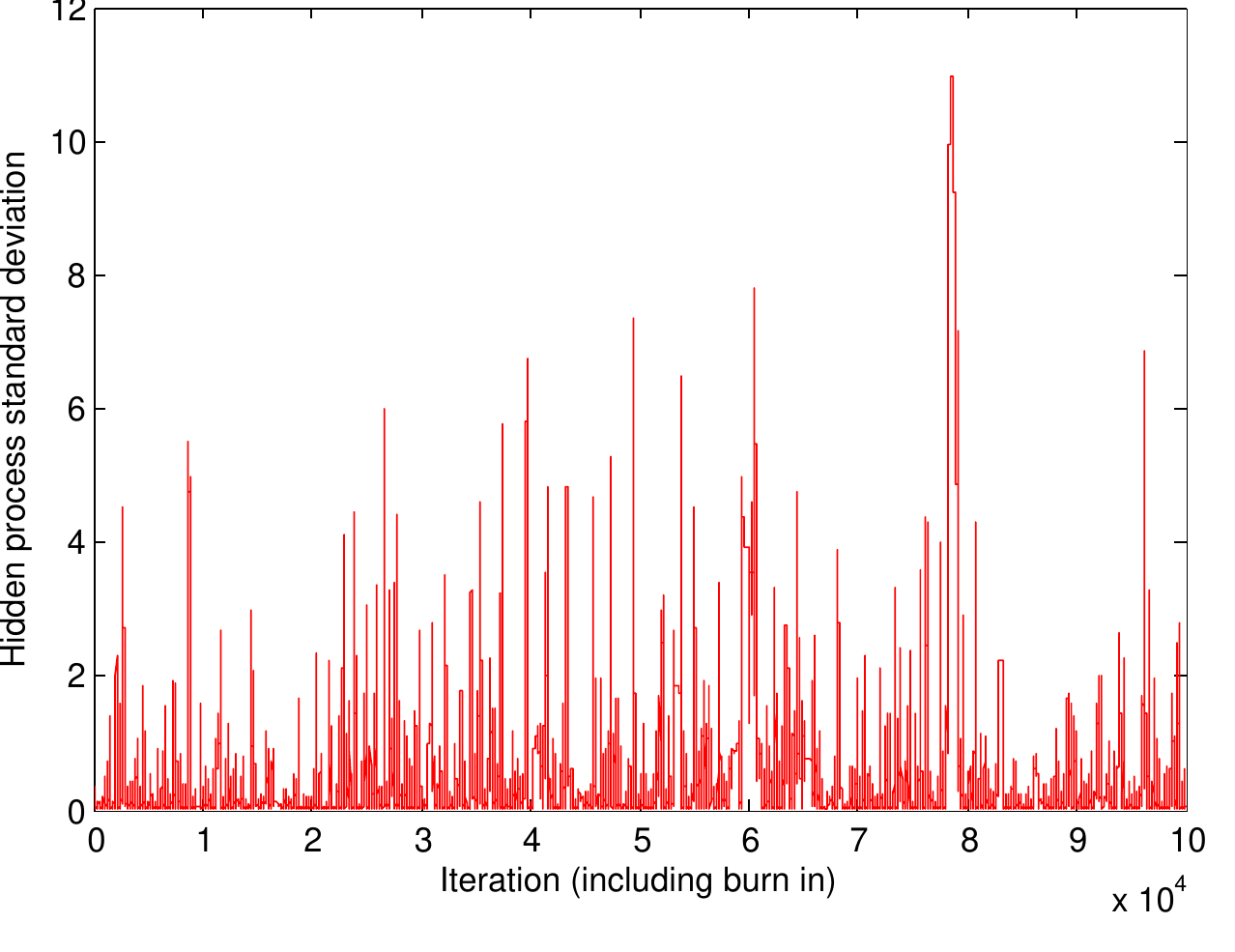}}
\scalebox{0.3}{\includegraphics[trim = 40mm 90mm 40mm 107mm, clip]{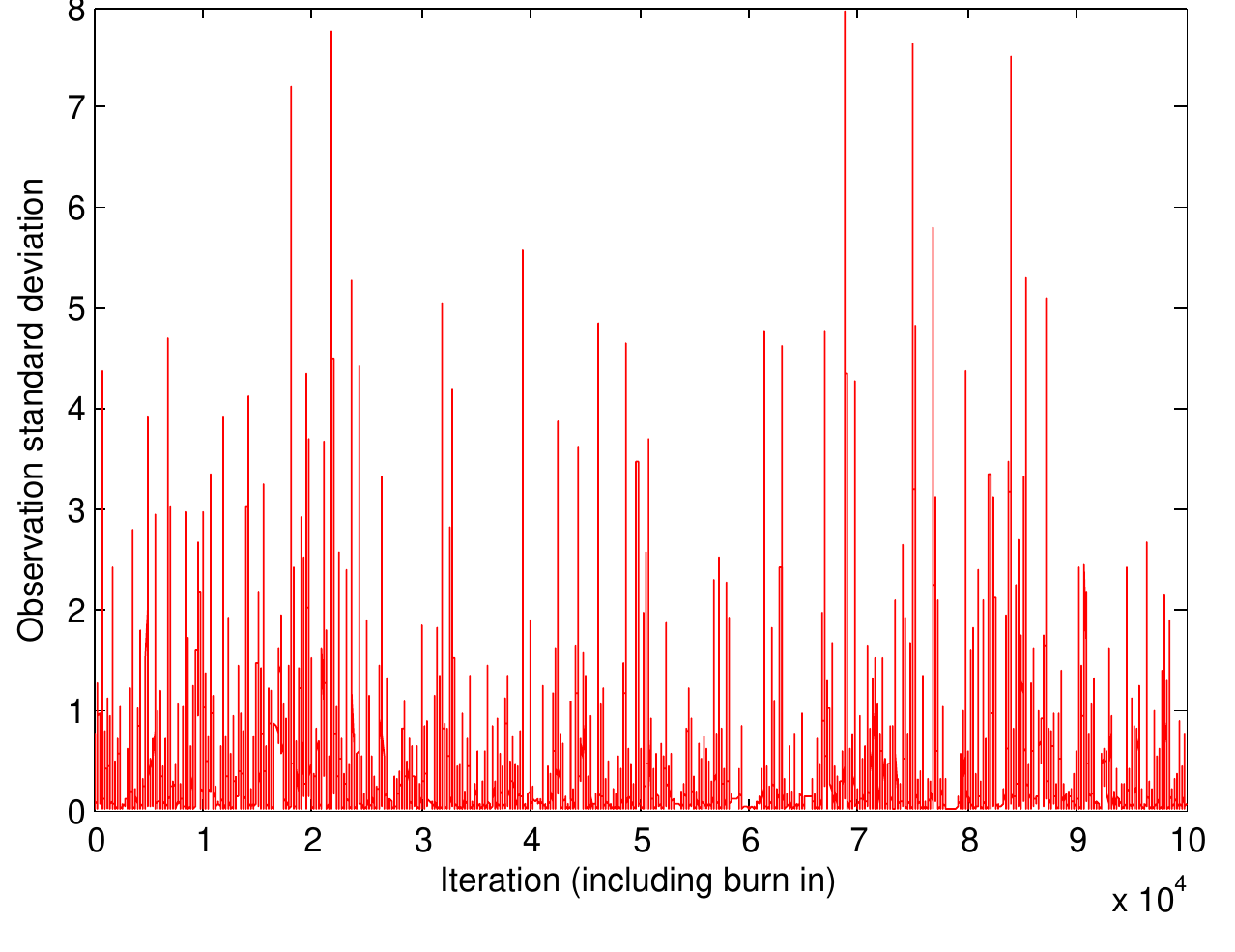}}
\caption[$\quad$Comparing alive PMMH to alive twisted PMMH]{Results for $N=100$.  Top: alive PMMH; bottom: alive twisted PMMH.  Trace plots for $F$, $\nu^2$, and $\gamma$ (from left to right).
Only one of the five repetitions of each simulation is shown.}
\label{fig:nocountryforoldmen8} 
\end{figure}
\par

\setcounter{chapter}{10}
\setcounter{section}{10}
\setcounter{subsection}{0}
\setcounter{figure}{0}
\setcounter{equation}{0} 
\noindent {\bf 10. Discussion}

In this paper, we introduced a change of measure on the alive particle filter of \cite{Jasra_2013} to reduce the variance of its estimate of the normalising constant.  By adopting similar assumptions as in \cite{Whiteley_2013}, we followed the theoretical framework developed in \cite{Whiteley_2013} to determine the unique, optimal change of measure for the alive algorithm.  That optimal choice also happens to be the same unique choice discovered in \cite{Whiteley_2013}, which, unlike this paper, does not consider HMMs whose observations have unknown or intractable likelihood densities.

We used our theoretical findings to formalise an alive twisted particle filter and an alive twisted PMMH.   Both methods were implemented on HMMs, with the PMMH being used in a real world example.  The numerical analyses illustrated that when the change of measure on the alive algorithms is not a close approximation of the ideal change in measure, twisting may not be worthwhile.  However, when a good approximation of the ideal change in measure was available, our algorithms did exhibit superior performance in some scenarios.

The assumptions used to prove our theoretical results may be difficult to verify.  Assumption (B\ref{hyp:H1anal}), in particular, would be very hard to verify when one knows little about the process producing the observations.  A future work might consider proving the same results under weaker assumptions, although that will likely not be a straightforward task.  Additionally, other future work might investigate possible applications outside of HMMs, such as in the rare events literature \cite{Cerou_2012} or ABC approximations of epidemiological models \cite{Del_Moral_2012}.
\par


\noindent {\large\bf Acknowledgment}
The first author was partially supported through a Roth Studentship at Imperial College London.  Both authors were also supported by an MOE Singapore grant.
\par


\begin{appendix}
\setcounter{chapter}{1}
\renewcommand{\theequation}{A.\arabic{equation}}
\setcounter{equation}{0} 
\noindent {\bf Appendix}

\noindent {\bf A. Proof of the main result from Section 6}

We first illustrate that ${Z}_{\theta,1:n}$ is actually finite in the limit as $n\rightarrow\infty$, for otherwise there would be no circumstance under which $\Upsilon\left(\widetilde{\mathbf{M}}_\theta\right)=0$.  The result is established as the following proposition.
 \begin{prop}\label{theo:analoguetoProp1NWAL}
Assume (B\ref{hyp:H1anal}) and (B\ref{hyp:H2anal}).  There exists a finite, real-valued constant $\Lambda$, which is independent of the initial distribution $\mu$, such that
$$
\frac{1}{n}\log\mu Q_{\theta,n}^\omega\left(1\right) \rightarrow \Lambda
$$
as $n\rightarrow\infty$, $\mathbb{P}-$almost surely.
\end{prop}

\begin{proof}[Proof of Proposition \ref{theo:analoguetoProp1NWAL}]
This proof closely follows the proof of \cite[Proposition 1]{Whiteley_2013}, with only some minor modifications.  Assuming (B\ref{hyp:H2anal}), we can define a constant $\underline{g}=\inf_{(\omega,x)} \nu Q_\theta^\omega\left( x \right)>0$ which also must be finite. Consider a sequence of random variables $\{ \kappa_n^\omega \}_{n \geq 1}$ where $\kappa_n^\omega=\nu Q_{\theta,n-1}^\omega\left(1\right) \underline{g}$.
 We know $\kappa_n^\omega>0$, and as
 \begin{align*}
 \kappa_{n+p}^\omega = \nu Q_{\theta,p+n-1}^\omega \left(1\right) \underline{g}
 = \nu Q_{\theta,p}^\omega Q_{\theta,n-1}^{z^{p}\omega} \left(1\right) \underline{g}
 = \nu Q_{\theta,p-1}^\omega Q_{\theta}^{z^{p-1}\omega} Q_{\theta,n-1}^{z^{p}\omega} \left(1\right) \underline{g} \geq \kappa_p^\omega \kappa_n^{z^p \omega},
 \end{align*}
 we have
 \begin{equation}\label{eq:someevidence2}
  -\log \kappa_{n+p}^\omega \leq -\log \kappa_p^\omega - \log \kappa_n^{z^p \omega}.
 \end{equation}
 Furthermore, (B\ref{hyp:H2anal}) and the definition of $\underline{g}$ ensure that each $\kappa_n$ is finite, and so
\begin{equation}\label{eq:someevidence3}
 \int_\Omega -\log \kappa_n^\omega \mathbb{P}(d\omega) > -\infty.
\end{equation}
Considering \eqref{eq:someevidence2}, \eqref{eq:someevidence3}, and the ergodicity of the shift operator assumed by (B\ref{hyp:H1anal}), we can apply Kingman's subbadditive ergodic theory \cite{Kingman_1976theorem} to obtain
\begin{equation}\label{eq:someevidence4}
 \frac{1}{n} \log \kappa_n^\omega \rightarrow \Lambda,
\end{equation}
as $n\rightarrow\infty$, $\mathbb{P}-$almost surely, where $\Lambda$ is a finite constant.

As (B\ref{hyp:H2anal}) implies
\begin{equation*}
 0 < \frac{\kappa_n^\omega}{\mu Q_{\theta,n}^\omega\left(1\right) }
 \leq \frac{\nu Q_{\theta,n}^\omega\left(1\right) }{\mu Q_{\theta,n}^\omega\left(1\right) } = \frac{\nu Q_{\theta}^\omega Q_{\theta,n-1}^{z\omega}\left(1\right) }{\mu Q_{\theta}^\omega Q_{\theta,n-1}^{z\omega}\left(1\right) } \leq \Delta_2, 
\end{equation*}
we have
\begin{equation}\label{eq:someevidence5}
 \sup_{\omega\in\Omega}\bigg| \frac{1}{n}\log \kappa_n^\omega - \frac{1}{n}\log\mu Q_{\theta,n}^\omega\left(1\right)\bigg|\leq \frac{1}{n}\log\Delta_2.
\end{equation}
Considering \eqref{eq:someevidence4} and \eqref{eq:someevidence5}, we find $(1/n)\log\mu Q_{\theta,n}^\omega\left(1\right) \rightarrow \Lambda$, as $n\rightarrow\infty$, $\mathbb{P}-$almost surely.
\end{proof}

The next two propositions clearly define the triple $\left(\eta,h,\lambda\right)$ that uniquely satisfies the system of equations \eqref{eq:ihavetoleavesoonforcoventgarden1558}.  It is not yet shown that $h$ is the optimal measure by which Algorithm \ref{alg:SMCalivetwistedsunday} should be twisted.  Before proving that, we have to first show that the measure exists.  The fact that the triple $\left(\eta,h,\lambda\right)$ uniquely satisfies \eqref{eq:ihavetoleavesoonforcoventgarden1558} is used in calculations in later parts of the proof of the main result.

\begin{prop}\label{theo:analoguetoProp2NWAL}
Assume (B\ref{hyp:H2anal}).
\begin{enumerate}
 \item{Fixing $\sigma\in\mathscr{P}(E)$, the limits
 \begin{align*}
 \eta^\omega (A) = \lim_{n \to \infty} \Phi_{\theta,n}^{z^{-n}\omega} \left( \sigma \right) \left( A \right)\quad \text{and} \quad
  h\left(\omega,x\right) = \lim_{n \to \infty} \frac{Q_{\theta,n}^\omega\left(1\right)\left(x\right)}{\Phi_{\theta,n}^{z^{-n}\omega} \left( \sigma \right)Q_{\theta,n}^\omega\left(1\right)}
 \end{align*}
exist, where $\eta^\omega (A)$ is a member of a family of probability measures, $\eta=\{\eta^\omega \in \mathscr{P}(E);\omega\in\Omega\}$, and $h\left(\omega,x\right)$ is a member of a family of real-valued, $\mathscr{F}\otimes\mathcal{E}-$measurable functions, $h:\Omega \times E \rightarrow \mathbb{R}$.}
\item{The families of probability measures and measurable functions just defined are independent of $\sigma$, and there exist constants $C<\infty$ and $\rho<1$ such that $\sup_{\omega\in \Omega} \sup_{\sigma\in \mathscr{P}(E)} \bigg| \bigg[\Phi_{\theta,n}^{z^{-n}\omega}\left( \sigma \right)-\eta^\omega\bigg]\left(\varphi\right) \bigg| \leq \sup_{x}|\varphi\left(x\right)| C \rho^n$ and

$\sup_{\omega\in \Omega} \sup_{x\in E} \sup_{\sigma\in \mathscr{P}(E)} \bigg| 
 \frac{Q_{\theta,n}^\omega\left(1\right)\left(x\right)}{\Phi_{\theta,n}^{z^{-n}\omega} \left( \sigma \right)Q_{\theta,n}^\omega\left(1\right)}-
h\left(\omega,x\right) \bigg| \leq C \rho^n$ for $\varphi\in \mathcal{B}_b (E)$ and $n\geq 1$.}
\item{The function $\lambda:\omega\in\Omega \rightarrow \eta^\omega \left(W^{\omega}\right)$ is $\mathscr{F}-$measurable, and
\begin{equation*}
 \sup_{(\omega,\omega)^{'}\in\Omega^2} \frac{\lambda_\omega}{\lambda_{\omega^{'}}}<\infty, \quad 
 \sup_{(\omega,\omega^{'},x,x^{'})\in\Omega^2 \times E^2} \frac{h\left(\omega,x\right)}{h\left(\omega^{'},x^{'}\right)}<\infty.
\end{equation*}
}
\item{Consider the triples that consist of (a) a family of probability measures on $\left(E,\mathcal{E} \right)$ indexed by $\Omega$, (b) an $\mathbb{R}^{+}-$valued measurable function on $\Omega \times E$, and (c) a measurable function on $\Omega$.  For all $\omega\in\Omega$, the triple $\left(\eta,h,\lambda\right)$ uniquely satisfies $\eta^\omega Q_\theta^\omega\left(\cdot\right) = \lambda_\omega \eta^{z\omega} \left(\cdot\right)$, $Q_\theta^\omega\left( h\left(z\omega,\cdot\right) \right)\left(x\right) = \lambda_\omega h\left(\omega,x\right)$, and $\eta^\omega (h\left(\omega,x\right)) = 1$.}
\end{enumerate}
\end{prop}

\begin{proof}[Proof of Proposition \ref{theo:analoguetoProp2NWAL}]
The proof is the same as the proof of \cite[Proposition 2]{Whiteley_2013}.  Even though our assumption (B\ref{hyp:H2anal}) differs slightly from the analogue in \cite{Whiteley_2013}, the necessary implications are the same.
\end{proof}

\begin{prop}\label{theo:analoguetoProp3NWAL}
Assume (B\ref{hyp:H1anal}) and (B\ref{hyp:H2anal}).  Then for $\Lambda$ as in Proposition \ref{theo:analoguetoProp1NWAL} and $\lambda$ as in Proposition \ref{theo:analoguetoProp2NWAL}, we have $\Lambda=\mathbb{E}[ \log \lambda ] = \int_\Omega \frac{Q_\theta^\omega\left(h\left(z\omega,\cdot\right)\right)\left(x\right)}{h\left(\omega,x\right)} \mathbb{P}\left(d\omega\right)$, for any $x\in E$.
\end{prop}

\begin{proof}[Proof of Proposition \ref{theo:analoguetoProp3NWAL}]
The proof is the same as the proof of \cite[Proposition 3]{Whiteley_2013}.  It only relies on the assumptions (B\ref{hyp:H1anal}) and (B\ref{hyp:H2anal}) because it makes use of Propositions \ref{theo:analoguetoProp1NWAL} and \ref{theo:analoguetoProp2NWAL}.
\end{proof}

Now that the triple $\left(\eta,h,\lambda\right)$ is clearly defined and it is known that $\widehat{Z}_{\theta,1:n}$ is approximating a finite value in the limit as $n\rightarrow\infty$, we can begin to establish how the optimal $h$ affects the particle filter.  That illustration begins by showing how the particle filter behaves when its transition density is any $\widetilde{\mathbf{M}}_\theta$ (which is a member of $\mathbb{M}$) and not necessarily one twisted with $h$.

The functions $\mathbf{J}_\theta$ and $\mathbf{L}_\theta$ of Definition \ref{def:coolM} can be used to construct \eqref{eq:Wewillfrequentlyrefertotheratiosgthai}.  Thus, for any $\widetilde{\mathbf{M}}_\theta$, we establish bounds on those functions via Lemmas \ref{theo:analoguetoLem1NWAL} and \ref{theo:analoguetoLem6NWAL} below to show that $(1/n)\log \widetilde{\mathcal{V}}_{\theta,n}^\omega \rightarrow \Upsilon\left(\widetilde{\mathbf{M}}_\theta\right)
 \quad \text{as} \quad n \rightarrow \infty$, $\mathbb{P}-$almost surely, in Proposition \ref{theo:analoguetoProp4NWAL} below.

\begin{lemma}\label{theo:analoguetoLem1NWAL}
Assume (B\ref{hyp:H2anal}) and (B\ref{hyp:H3}).  For all $\omega,\omega^{'}\in\Omega$, $x,x^{'}\in E$, $T_{\omega}\geq N$, and $T_{\omega^{'}}\geq N$, ${\mathbf{W}^{T_\omega-1}\left(\omega,x\right)}/{\mathbf{W}^{T_{\omega^{'}}-1}\left(\omega^{'},x^{'}\right)} \leq \Delta_3$ for some $\Delta_3\in(0,\infty)$.  Furthermore, $\epsilon_{-}^{T_{z\omega}-1}\nu^{\otimes\left(T_{z\omega}-1\right)}\left(\cdot\right) \leq \mathbf{M}^{T_{\omega}-1,T_{z\omega}-1}_\theta\left(\omega, x,\cdot\right) \leq \epsilon_{+}^{T_{z\omega}-1}\nu^{\otimes\left(T_{z\omega}-1\right)}\left(\cdot\right)$.
\end{lemma}

\begin{proof}[Proof of Lemma \ref{theo:analoguetoLem1NWAL}]
Under (B\ref{hyp:H3}), it is clear that any ${\mathbf{W}^{T_\omega-1}\left(\omega,x\right)}$ is positive and finite, and so a positive and finite upper bound on ${\mathbf{W}^{T_\omega-1}\left(\omega,x\right)}/{\mathbf{W}^{T_{\omega^{'}}-1}\left(\omega^{'},x^{'}\right)}$ is obvious.

Recalling (B\ref{hyp:H2anal}) and \eqref{eq:Mm1m2defnsgtues}, we can calculate
\begin{align*}
 \mathbf{M}^{T_{\omega}-1,T_{z\omega}-1}_\theta\left(\omega, x,{d}x\right)
 &\leq \prod_{i=1}^{T_{z\omega}-1} \frac{\left( N-1 \right) \epsilon_{+}\nu\left({d}x^i\right) }{N-1}
 = \epsilon_{+}^{T_{z\omega}-1}\nu^{\otimes\left(T_{z\omega}-1\right)}\left({d}x\right).
\end{align*}
Similarly, we have the lower bound $\mathbf{M}^{T_{\omega}-1,T_{z\omega}-1}_\theta\left(\omega, x,{d}x\right) \geq \epsilon_{-}^{T_{z\omega}-1}\nu^{\otimes\left(T_{z\omega}-1\right)}\left({d}x\right)$.
\end{proof}

\begin{lemma}\label{theo:analoguetoLem6NWAL}
For any $\omega\in\Omega$, assume (B\ref{hyp:H2anal}) and (B\ref{hyp:H3}), let $\widetilde{\mathbf{M}}^{T_{\omega}-1,T_{z\omega}-1}_\theta$ be any member of $\mathbb{M}^{T_{\omega}-1,T_{z\omega}-1}$, and let $\widetilde{\nu}$ be as in the definition of $\mathbb{M}^{T_{\omega}-1,T_{z\omega}-1}$.  There exist constants $\alpha\in(0,\infty)$ and $\left(\delta_{-},\delta_{+}\right)\in\left(0,\infty\right)^2$ and a probability measure $\sigma\in\mathscr{P}(E^{T_{z\omega}-1})$ such that ${\mathbf{J}^{T_{\omega}-1,T_{z\omega}-1}_\theta\left(\omega, x\right)}/{\mathbf{J}^{T_{\omega^{'}}-1,T_{z\omega^{'}}-1}_\theta\left(\omega^{'}, x^{'}\right)} \leq \alpha$ for all $\left(\omega, \omega^{'}, x, x^{'}\right)\in \Omega^2 \times E^{T_{\omega}-1} \times E^{T_{\omega^{'}}-1}$ and $\delta_{-}\sigma\left(\cdot\right) \leq \mathbf{L}^{T_{\omega}-1,T_{z\omega}-1}_\theta\left(\omega, x,\cdot\right) \leq \delta_{+}\sigma\left(\cdot\right)$ for all $\left(\omega, x\right)\in \Omega \times E^{T_{\omega}-1}$, where $\sigma\left( dx\right)\propto \left( \left({\mathrm{d}\nu^{\otimes (T_{z\omega}-1)}}/{\mathrm{d}\widetilde{\nu}}\right) \left(x\right)\right)^2 \widetilde{\nu}\left(dx\right)$.
\end{lemma}

\begin{proof}[Proof of Lemma \ref{theo:analoguetoLem6NWAL}]
For any $A\in \mathcal{E}^{\otimes (T_{z\omega}-1)}$,
\begin{align}\label{eq:numeratorJJsgthai}
 \int_{A} &\mathbf{W}^{T_{\omega}-1}\left(\omega,x\right)^2 \phi_\theta^{\omega,T_{\omega}-1,T_{z\omega}-1}\left(x,x^{'}\right)^2 \widetilde{\mathbf{M}}^{T_{\omega}-1,T_{z\omega}-1}_\theta\left(\omega, x,{d}x^{'}\right)\\ \nonumber
 &\leq \sup_{\left(\omega^{'},z\right)\in\Omega\times E^{T_{\omega^{'}}-1}} \mathbf{W}^{T_{\omega^{'}}-1}\left(\omega^{'},z\right)^2  \left(\Delta_3 \frac{\epsilon_{+}^{T_{z\omega}-1}}{\widetilde{\epsilon}_{-}}\right)^2 \widetilde{\epsilon}_{+} \int_{A} \left(\frac{\mathrm{d}\nu^{\otimes (T_{z\omega}-1)}}{\mathrm{d}\widetilde{\nu}}\left(x^{'}\right)\right)^2\widetilde{\nu}\left(dx^{'}\right) \\ \nonumber &< \infty,
\end{align}
by Lemma \ref{theo:analoguetoLem1NWAL} and \eqref{eq:Mm1m2defnsgtues}.  Similarly,
\begin{align}\label{eq:denominatorJJsgthai}
 \int_{A} &\mathbf{W}^{T_{\omega}-1}\left(\omega,x\right)^2 \phi_\theta^{\omega,T_{\omega}-1,T_{z\omega}-1}\left(x,x^{'}\right)^2 \widetilde{\mathbf{M}}^{T_{\omega}-1,T_{z\omega}-1}_\theta\left(\omega, x,{d}x^{'}\right)\\ \nonumber
 &\geq \inf_{\left(\omega^{'},z\right)\in\Omega\times E^{T_{\omega^{'}}-1}} \mathbf{W}^{T_{\omega^{'}}-1}\left(\omega^{'},z\right)^2  \left(\frac{\epsilon_{-}^{T_{z\omega}-1}}{\Delta_3 \widetilde{\epsilon}_{+}}\right)^2 \widetilde{\epsilon}_{-} \int_{A} \left(\frac{\mathrm{d}\nu^{\otimes (T_{z\omega}-1)}}{\mathrm{d}\widetilde{\nu}}\left(x^{'}\right)\right)^2\widetilde{\nu}\left(dx^{'}\right),
\end{align}
by Lemma \ref{theo:analoguetoLem1NWAL} and \eqref{eq:Mm1m2defnsgtues}.  Taking $A=E^{T_{z\omega}-1}$ and dividing \eqref{eq:numeratorJJsgthai} by \eqref{eq:denominatorJJsgthai}, we have 
\begin{equation*}
  \frac{\mathbf{J}^{T_{\omega}-1,T_{z\omega}-1}_\theta\left(\omega, x\right)}{\mathbf{J}^{T_{\omega^{'}}-1,T_{z\omega^{'}}-1}_\theta\left(\omega^{'}, x^{'}\right)}\leq
  \left((\Delta_3)^3 \frac{\widetilde{\epsilon}_{+} \epsilon_{+}^{T_{z\omega}-1}}{\widetilde{\epsilon}_{-} \epsilon_{-}^{T_{z\omega}-1}}\right)^2 \frac{\widetilde{\epsilon}_{+}}{\widetilde{\epsilon}_{-}}=\alpha.
\end{equation*}
Finally, as $\mathbf{L}^{T_{\omega}-1,T_{z\omega}-1}_\theta\left(\omega, x,\cdot\right)$ is a kernel that has $\mathbf{J}^{T_{\omega}-1,T_{z\omega}-1}_\theta\left(\omega, x\right)$ as a normalising constant, it is clear via \eqref{eq:numeratorJJsgthai} and \eqref{eq:denominatorJJsgthai} that
\begin{equation*}
 \delta_{-}=\inf_{\left(\omega^{'},z\right)\in\Omega\times E^{T_{\omega^{'}}-1}} \mathbf{W}^{T_{\omega^{'}}-1}\left(\omega^{'},z\right)^2  \left(\frac{\epsilon_{-}^{T_{z\omega}-1}}{\Delta_3 \widetilde{\epsilon}_{+}}\right)^2 \widetilde{\epsilon}_{-} \left(\int_{E} \left(\frac{\mathrm{d}\nu^{\otimes (T_{z\omega}-1)}}{\mathrm{d}\widetilde{\nu}}\left(x^{'}\right)\right)^2\widetilde{\nu}\left(dx^{'}\right)\right)^{-1}
\end{equation*}
and
\begin{equation*}
 \delta_{+}=\sup_{\left(\omega^{'},z\right)\in\Omega\times E^{T_{\omega^{'}}-1}} \mathbf{W}^{T_{\omega^{'}}-1}\left(\omega^{'},z\right)^2  \left(\Delta_3 \frac{\epsilon_{+}^{T_{z\omega}-1}}{\widetilde{\epsilon}_{-}}\right)^2 \widetilde{\epsilon}_{+} \left(\int_{E} \left(\frac{\mathrm{d}\nu^{\otimes (T_{z\omega}-1)}}{\mathrm{d}\widetilde{\nu}}\left(x^{'}\right)\right)^2\widetilde{\nu}\left(dx^{'}\right)\right)^{-1}
\end{equation*}
when $\sigma\left( dx\right)\propto \left( \left({\mathrm{d}\nu^{\otimes (T_{z\omega}-1)}}/{\mathrm{d}\widetilde{\nu}} \right) \left(x\right)\right)^2 \widetilde{\nu}\left(dx\right)$.
\end{proof}

\begin{prop}\label{theo:analoguetoProp4NWAL}
Assume (B\ref{hyp:H1anal}), (B\ref{hyp:H2anal}) and (B\ref{hyp:H3}). For each $\widetilde{\mathbf{M}}^{T_{\omega}-1,T_{z\omega}-1}_\theta$ any member of a $\mathbb{M}^{T_{\omega}-1,T_{z\omega}-1}$, there exists a non-negative, finite constant $\Upsilon\left(\widetilde{\mathbf{M}}_\theta\right)$, which is independent of the initial distribution $\mu$, such that $(1/n)\log \widetilde{\mathcal{V}}_{\theta,n}^\omega \rightarrow \Upsilon\left(\widetilde{\mathbf{M}}_\theta\right)$ as $n\rightarrow\infty$, $\mathbb{P}-$almost surely.
\end{prop}

\begin{proof}[Proof of Proposition \ref{theo:analoguetoProp4NWAL}]
Assume (B\ref{hyp:H1anal}) and (B\ref{hyp:H2anal}).  By Proposition \ref{theo:analoguetoProp1NWAL}, for any $\mu\in \mathscr{P}(E)$, $(2/n)\log\mu Q_{\theta,n}^\omega\left(1\right) \rightarrow 2\Lambda$ as $n\rightarrow\infty$, $\mathbb{P}-$almost surely.

Next, assume (B\ref{hyp:H3}) and consider the bounds presented in Lemma \ref{theo:analoguetoLem6NWAL}.  Following the exact same steps as in the proof of \cite[Proposition 1]{Whiteley_2013}, one can show that there exists a constant $\Xi\in\left(-\infty,\infty\right)$ such that $(1/n)\log\mu^{\otimes(T_{\omega}-1)}\mathbf{R}_{\theta,n}^{\omega,T_{\omega}-1,T_{z^n\omega}-1}\left(1\right)$ approaches $\Xi$ as $n\rightarrow\infty$, $\mathbb{P}-$almost surely.

By the definition \eqref{eq:Wewillfrequentlyrefertotheratiosgthai}, we have $\widetilde{\mathcal{V}}_{\theta,n}^\omega = {\mu^{\otimes(T_{\omega}-1)}\mathbf{R}_{\theta,n}^{\omega,T_{\omega}-1,T_{z^n \omega}-1}\left(1\right)}/{\left(\mu Q_{\theta,n}^\omega\left(1\right)\right)^2}$, and so $\Upsilon\left(\widetilde{\mathbf{M}}_\theta\right)=\Xi-2\Lambda$.
\end{proof}

There is one final lemma which is needed to prove Theorem \ref{theo:analoguetoTheorem1NWAL}.  The following is a technical result establishing that an additive functional of the form \eqref{eq:labeltheaddfunclh}, with optimal $h$ as defined in Proposition \ref{theo:analoguetoProp2NWAL}, is an eigenfunction for $\mathbf{Q}_\theta$.

\begin{lemma}\label{theo:analoguetoLem2NWAL}
Assume (B\ref{hyp:H2anal}). Then for any $\omega\in\Omega$, $\mathbf{Q}_\theta^{\omega,m_1,m_2}\left(\mathbf{h}^{m_2}\left(z\omega,\cdot\right)\right)\left(x\right)= \lambda_\omega \mathbf{h}^{m_1}\left(\omega,x\right)$, where $\lambda_\omega$ is as in Proposition \ref{theo:analoguetoProp2NWAL}.
\end{lemma}

\begin{proof}[Proof of Lemma \ref{theo:analoguetoLem2NWAL}]
\begin{align*}
 &\mathbf{Q}_\theta^{\omega,m_1,m_2}\left(\mathbf{h}^{m_2}\left(z\omega,\cdot\right)\right)\left(x\right)=\int_{E^{m_2}} \mathbf{Q}^{m_1,m_2}_\theta \left(\omega, x,{d}u\right)\mathbf{h}^{m_2}\left(z\omega,u\right) \\
 &=\frac{1}{m_2} \sum_{j=1}^{m_2} \mathbf{W}^{m_1}\left(\omega,x\right) \int_{E} \frac{\sum_{i=1}^{m_1} W\left(\omega,x^i\right)M_\theta\left(\omega, x^i,{d}u^j\right)}{\sum_{l=1}^{m_1} W\left(\omega,x^l\right)} h\left(z\omega,u^j\right) \\
 &=\frac{1}{m_1} \frac{1}{m_2} \sum_{j=1}^{m_2} \sum_{i=1}^{m_1} \lambda_\omega h\left(\omega,x^i\right)
 = \lambda_\omega \frac{1}{m_1} \sum_{i=1}^{m_1} h\left(\omega,x^i\right)
 = \lambda_\omega \mathbf{h}^{m_1}\left(\omega,x\right),
\end{align*}
where we have applied Proposition \ref{theo:analoguetoProp2NWAL}.
\end{proof}

Finally, we use the optimal $h$ of Proposition \ref{theo:analoguetoProp2NWAL} and prove that the specific, unique $\widetilde{\mathbf{M}}_\theta$ defined in Theorem \ref{theo:analoguetoTheorem1NWAL} achieves $\Upsilon\left(\widetilde{\mathbf{M}}_\theta\right)=0$.  The main result is presented as Theorem \ref{theo:analoguetoTheorem1NWAL} in Section 6, and its proof follows the same steps as in the proof of \cite[Theorem 1]{Whiteley_2013}.

\setcounter{chapter}{2}
\noindent {\bf B. Tables}

\begin{table}[H]
\caption{Distributions used throughout}
\centering
\begin{tabular}{l l l l}
\hline
\hline
Distribution & Parameters & Notation & Expected value \\[0.5ex]
\hline
Gamma & shape $\alpha>0$ and scale $\beta>0$ & $\mathcal{G}a\left( \alpha, \beta \right)$ & $\alpha \beta$
\\ [2ex]
Gaussian & mean $\mu\in(-\infty,\infty)$ and variance $\sigma^2>0$ & $\mathcal{N}\left( \mu, \sigma^2 \right)$ & $\mu$
\\ [2ex]
Stable & stability $\alpha\in(0,2]$, skewness $\beta\in[-1,1]$, & $\mathcal{S}\left(\alpha,\beta,\gamma,\delta\right)$ & $\mu$ when $\alpha>1$
\\
 & scale $\gamma>0$, and location $\delta\in(-\infty,\infty)$ &  &
\\ [1ex]
\hline
\end{tabular}
\label{table:dist}
\end{table}

\begin{table}[H]
\caption{Kernel and operator notation used throughout}
\centering
\begin{tabular}{l l}
\hline
\hline
 & Kernels \\[0.5ex]
\hline
1. & $Q_\theta\left(\omega, x,{d}x\right) = W\left(\omega,x\right) M_\theta\left(\omega, x,{d}x\right)$
\\[1.25ex]
2a. & $Q_{\theta,n}\left(\omega, x,x^{'}\right) = \int Q_{\theta,n-1} \left(\omega, x,u\right) Q_\theta\left(z^{n-1}\omega, du,x^{'}\right), \quad n \geq 1$ 
\\[1.25ex]
2b. & $Q_{\theta,n}\left(\omega, x,x^{'}\right) = Q_{\theta,n-1}^\omega Q_\theta\left(z^{n-1}\omega, \cdot,x^{'}\right), \quad n \geq 1$, with $Q_{\theta,0}\left(\omega, x,x\right) = \text{I}$
\\[1.25ex]
2c. & $Q_{\theta,p+n}\left(\omega, x,{d}x\right) = Q_{\theta,p}^\omega Q_{\theta,n}\left(z^{p}\omega, \cdot,{d}x\right)$ $\quad-$ via induction as in \cite{Whiteley_2013}
\\[1.25ex]
3. & $\mathbf{Q}^{m_1,m_2}_\theta\left(\omega, x,{d}x\right) = \mathbf{W}^{m_1}\left(\omega,x\right) \mathbf{M}^{m_1,m_2}_\theta\left(\omega, x,{d}x\right)$
\\[1ex]
\hline
 & Operators \\[0.5ex]
\hline
1. & $M_\theta^\omega\left( W \right)\left( x \right)=\int_E M_\theta\left(\omega, x,{d}x^{'}\right) W\left(z\omega,x^{'}\right)$
\\[1.25ex]
2a. & $Q_\theta^\omega\left(\varphi\right)\left(x\right)=\int_E Q_\theta\left(\omega, x,{d}x^{'}\right)\varphi\left(x^{'}\right), \quad \varphi\in \mathcal{B}_b (E)$
\\[1.25ex]
2b. & $Q_{\theta,n}^\omega\left(\varphi\right)\left(x\right)=\int_{E^n} Q_{\theta,n}\left(\omega, x,{d}x^{'}\right)\varphi\left(x^{'}\right), \quad \varphi\in \mathcal{B}_b (E)$
\\[1.25ex]
3a. & $\sigma Q_\theta^\omega\left(\cdot\right)=\int_E Q_\theta\left(\omega, x,\cdot\right)\sigma\left(dx\right), \quad \sigma\in \mathscr{M}(E)$ or $\sigma\in\mathscr{P}(E)$
\\[1.25ex]
3b. & $\sigma Q_{\theta,n}^\omega\left(\cdot\right)=\int_{E^n} Q_{\theta,n}\left(\omega, x,\cdot\right)\sigma\left(dx\right), \quad \sigma\in \mathscr{M}(E)$ or $\sigma\in\mathscr{P}(E)$
\\[1.25ex]
3c. & $\sigma Q_{\theta,n}^\omega\left(W^{z^n\omega}\right)=\int_{E^{n+1}} Q_{\theta,n}\left(\omega, x,x^{'}\right)W\left({z^n\omega},x^{'}\right)\sigma\left(dx^{'}\right)=\sigma Q_{\theta,n+1}^\omega\left(1\right)$
\\[1.25ex]
4a. & $\mathbf{Q}_\theta^{\omega,m_1,m_2}\left(\mathbf{f}^{m_2}\left(z\omega,\cdot\right)\right)\left(x\right)=\int_{E^{m_2}} \mathbf{Q}^{m_1,m_2}_\theta \left(\omega, x,{d}x^{'}\right)\mathbf{f}^{m_2}\left(z\omega,x^{'}\right), \quad f\in \mathcal{B}_b (E)$
\\[1.25ex]
4b. & $\sigma\mathbf{Q}_\theta^{\omega,m_1,m_2}\left(\cdot\right)=\int_{E^{m_1}} \mathbf{Q}^{m_1,m_2}_\theta \left(\omega, x,\cdot\right)\sigma\left(dx\right), \quad \sigma\in \mathscr{M}(E^{m_1})$ or $\sigma\in\mathscr{P}(E^{m_1})$
\\[1ex]
\hline
 & Probability measures \\[0.5ex]
\hline
1. & $\Phi_\theta^\omega\left( \sigma \right) \left(\cdot\right)= \frac{\sigma Q_\theta^\omega }{\sigma Q_\theta^\omega\left(1\right)}\left(\cdot\right)$
\\[1.25ex]
2. & $\Phi_{\theta,n}^\omega \left( \sigma \right) \left(\cdot\right)= \left( \Phi_\theta^{z^{n-1}\omega} \circ \Phi_{\theta,n-1}^\omega \right) \left( \sigma \right) \left(\cdot\right), \quad n \geq 1$, with $\Phi_{\theta,0}^\omega \left( \sigma \right) \left(\cdot\right)= \text{I}$
\\[1.25ex]
3. & $\Phi_{\theta,n}^\omega\left( \sigma \right) = \frac{\sigma Q_{\theta,n}^\omega }{\sigma Q_{\theta,n}^\omega\left(1\right)} = \left( \Phi_{\theta,n-1}^{z\omega} \circ \Phi_{\theta}^\omega \right) \left( \sigma \right)$ $\quad-$ via induction as in \cite{Whiteley_2013}
\\ [1ex]
\hline
\end{tabular}
\label{table:kernelnotation}
\end{table}

\end{appendix}



\bibliography{thesis_biblio}{}

\begin{thebibliography}{10}

\bibitem{Anderson_1979}
B.~D.~O. Anderson and J.~B. Moore.
\newblock {\em Optimal Filtering}.
\newblock Prentice--Hall, Englewood Cliffs, 1979.

\bibitem{Andrieu_2010}
C.~Andrieu, A.~Doucet, and R.~Holenstein.
\newblock Particle markov chain monte carlo methods (with discussion).
\newblock {\em Journal of the Royal Statistical Society: Series B},
  72:269--342, 2010.

\bibitem{Cappe_2005}
O.~Cappe, E.~Moulines, and T.~Ryden.
\newblock {\em Inference in Hidden Markov Models. Series: Statistics}.
\newblock Springer, New York, 2005.

\bibitem{Cerou_2012}
F.~{C'erou}, P.~{Del Moral}, T.~Furon, and A.~Guyader.
\newblock Sequential monte carlo for rare event estimation.
\newblock {\em Statistics and Computing}, 22:795--808, 2012.

\bibitem{Colella_2007}
S.~Colella, C.~Yau, J.~M. Taylor, G.~Mirza, H.~Butler, P.~Clouston, A.~S.
  Bassett, A.~Seller, C.~C. Holmes, and J.~Ragoussis.
\newblock Quantisnp: an objective bayes hidden-markov model to detect and
  accurately map copy number variation using snp genotyping data.
\newblock {\em Nucleic Acids Research}, 35(6):2013--2025, 2007.

\bibitem{Doucet_2008}
D.~Crisan and B.~Rozovsky (editors).
\newblock {\em Handbook of Nonlinear Filtering}.
\newblock Oxford University Press, Oxford, 2011.

\bibitem{Dean_2010}
T.~A. Dean, S.~S. Singh, A.~Jasra, and G.~W. Peters.
\newblock Parameter estimation for hidden markov models with intractable
  likelihoods.
\newblock Technical Report University of Cambridge, Department of Engineering,
  2010.

\bibitem{Del_Moral_2004}
P.~{Del Moral}.
\newblock {\em Feynman-Kac formulae. Genealogical and interacting particle
  approximations. Series: Probability and Applications}.
\newblock Springer-Verlag, Heidelberg, 2004.

\bibitem{Del_Moral_2012}
P.~{Del Moral}, A.~Doucet, and A.~Jasra.
\newblock An adaptive sequential monte carlo method for approximate bayesian
  computation.
\newblock {\em Statistics and Computing}, 22:1223--1237, 2012.

\bibitem{Doucet_2001}
A.~Doucet, S.~Godsill, and C.~Andrieu.
\newblock On sequential monte carlo sampling methods for bayesian filtering.
\newblock {\em Statistics and Computing}, 10:197--208, 2000.

\bibitem{Gordon_1993}
N.~J. Gordon, D.~J. Salmond, and A.~F.~M. Smith.
\newblock Novel approach to nonlinear/non-gaussian bayesian state estimation.
\newblock {\em IEE-Proceedings-F}, 140:107--113, 1993.

\bibitem{Hu_1996}
J.~Hu, M.~K. Brown, and W.~Turin.
\newblock Hmm based online handwriting recognition.
\newblock {\em IEEE Transactions on Pattern Analysis and Machine Intelligence},
  18(10):1039--1045, 1996.

\bibitem{Jasra_2013}
A.~Jasra, A.~Lee, C.~Yau, and X.~Zhang.
\newblock The alive particle filter. preprint.
\newblock (arXiv:1304.0151v1 [stat.CO]), 2013.

\bibitem{Jasra_2012b}
A.~Jasra, S.~S. Singh, J.~S. Martin, and E.~McCoy.
\newblock Filtering via approximate bayesian computation.
\newblock {\em Statistics and Computing}, 22(6):1223--1237, 2012.

\bibitem{Julier_2004}
S.~J. Julier and J.~K. Uhlmann.
\newblock Unscented filtering and nonlinear estimation.
\newblock {\em Proceedings of the IEEE}, 92(3):401--422, 2004.

\bibitem{Kingman_1976theorem}
J.~F.~C. Kingman.
\newblock Subadditive ergodic theory.
\newblock {\em Annals of Probability}, 1(6):883--909, 1976.

\bibitem{Langrock_2012}
R.~Langrock, I.~L. MacDonald, and W.~Zucchini.
\newblock Some nonstandard stochastic volatility models and their estimation
  using structured hidden markov models.
\newblock {\em Journal of Empirical Finance}, 19(1):147--161, 2012.

\bibitem{Martin_2012}
J.~S. Martin, A.~Jasra, S.~S. Singh, N.~Whiteley, and E.~McCoy.
\newblock Approximate bayesian computation for smoothing. preprint.
\newblock (arXiv:1206.5208v1 [stat.CO]), 2012.

\bibitem{Peshkin_1999}
L.~Peshkin and M.~S. Gelfand.
\newblock Segmentation of yeast dna using hidden markov models.
\newblock {\em Bioinformatics}, 15(12):980--986, 1999.

\bibitem{Rosales_2004}
R.~A. Rosales, M.~Fill, and A.~L. Escobar.
\newblock Calcium regulation of single ryanodine receptor channel gating
  analyzed using hmm/mcmc statistical methods.
\newblock {\em Journal of General Physiology}, 123:533--553, 2004.

\bibitem{Shephard_1997bb}
N.~Shephard and M.~K. Pitt.
\newblock Likelihood analysis of non-gaussian measurement time series.
\newblock {\em Biometrika}, 84(3):653--667, 1997.

\bibitem{Tavare_1997}
S.~Tavare, D.~J. Balding, R.~C. Griffiths, and P.~Donnelly.
\newblock Inferring coalescence times from dna sequence data.
\newblock {\em Genetics}, 145:505--518, 1997.

\bibitem{Whiteley_2013}
N.~Whiteley and A.~Lee.
\newblock Twisted particle filters. preprint.
\newblock \url{http://www.maths.bris.ac.uk/~manpw/}, 2013.

\bibitem{Yau_2011}
C.~Yau, O.~Papaspiliopoulos, G.~O. Roberts, and C.~C. Holmes.
\newblock Bayesian non-parametric hidden markov models with applications in
  genomics.
\newblock {\em Journal of the Royal Statistical Society: Series B},
  73(1):33--57, 2011.

\end{thebibliography}
\bibliographystyle{plain}


\vskip .65cm
\noindent
Department of Mathematics, Imperial College London, London, SW7 2AZ, UK.
\vskip 2pt
\noindent
E-mail: a.persing11@ic.ac.uk

\vskip .65cm
\noindent
Department of Statistics \& Applied Probability, National University of Singapore, Singapore, 117546, SG.
\vskip 2pt
\noindent
E-mail: staja@nus.edu.sg


\end{document}